\documentclass[preprint]{imsart}

\usepackage[utf8]{inputenc}
\usepackage[colorlinks,citecolor=blue,urlcolor=blue,linkcolor=magenta]{hyperref}
\usepackage[abs]{overpic}
\usepackage[authoryear]{natbib}
\usepackage{times}
\usepackage{multirow}
\usepackage{amssymb,amsmath, amsthm}
\usepackage{graphicx}
\usepackage{framed} %bar
\usepackage[english]{babel}
\usepackage{booktabs}
\usepackage{color}
\usepackage{epsfig}
\usepackage{paralist}
\usepackage{verbatim}
\usepackage{mathabx}  % for symbols on legends for figs

\newcommand{\si}{\sigma}

\renewcommand{\th}{\theta}
\newcommand{\la}{\lambda}
\newcommand{\ga}{\gamma}

\newcommand{\eps}{\varepsilon}

\renewcommand{\phi}{\varphi}

\newcommand{\scr}[1]{{\mathcal #1}}
\newcommand{\EE}{\mathbb{E}}

\newcommand{\PP}{\mathbb{P}}

\newcommand{\RR}{\mathbb{R}}

\renewcommand{\tilde}{\widetilde}

\providecommand{\I}{\mathrm{I}}

\providecommand{\diag}[0]{{\operatorname{diag}}}

\providecommand{\trace}{{\operatorname{tr}}}

\newcommand{\rara}[1]{\renewcommand{\arraystretch}{#1}}

\newcommand{\vt}[1]{\breve{#1}}

\mathchardef\given="626A

\providecommand{\e}{\mathrm{e}}

\newcommand{\bem}{\begin{bmatrix}}
\newcommand{\enm}{\end{bmatrix}}

\newcommand{\bb}[1]{\mathbb{#1}}

\newtheorem{thm}{Theorem}[section]
\newtheorem{lemma}[thm]{Lemma}
\newtheorem{cll}[thm]{Corollary}
\theoremstyle{definition}
\newtheorem{algorithm}{Algorithm}
\newtheorem{prop}[thm]{Proposition}
\newtheorem{rem}[thm]{Remark}
\newtheorem{ex}[thm]{Example}
\newtheorem{defn}[thm]{Definition}

\newcommand{\dd}{{\,\mathrm d}}

\newcommand{\expec}[1]{\ensuremath{{\rm E}\mspace{-1mu}\left[#1\right]}}

\newcommand{\T}{{\prime}}

\begin{document}

\begin{frontmatter}

\title{Bayesian estimation of  discretely observed  multi-dimensional diffusion processes using guided proposals}

\runtitle{Bayesian estimation for diffusions using guided proposals}

\begin{aug}
\author{Frank van der Meulen\ead[label=e1]{f.h.vandermeulen@tudelft.nl} and
Moritz Schauer, \ead[label=e2]{m.r.schauer@math.leidenuniv.nl}}

\runauthor{Van der Meulen and Schauer}

\affiliation{Delft University of Technology and  Leiden University }

\address{Delft Institute of Applied Mathematics (DIAM) \\
Delft University of Technology\\
Mekelweg 4\\
2628 CD Delft\\
The Netherlands\\
\printead{e1}\\[1em]
Mathematical Institute\\
Leiden University\\
P.O. Box 9512\\
2300 RA Leiden\\
The Netherlands\\
\printead{e2}}

\end{aug}

\begin{abstract}

Estimation of parameters of a diffusion based on discrete time observations poses a difficult problem due to the lack of a closed form expression for the likelihood. From a Bayesian computational perspective it can be casted as a missing data problem where the diffusion bridges in between discrete-time observations are missing. The computational problem can then be dealt with using a Markov-chain Monte-Carlo method known as  data-augmentation.
If unknown parameters appear in the diffusion coefficient, direct implementation of data-augmentation results in a Markov chain that is reducible. Furthermore, data-augmentation requires  efficient sampling of diffusion bridges, which can be difficult, especially in the multidimensional case.

We present a general framework to deal with with these problems that does not rely on discretisation.  The construction generalises previous approaches and sheds light on the assumptions necessary to make these approaches work. We define a random-walk type Metropolis-Hastings sampler for updating diffusion bridges. Our methods are illustrated  using guided proposals for sampling diffusion bridges. These are Markov processes obtained by adding a guiding term to the drift of the diffusion. We give  general guidelines on the construction of these proposals and  introduce a time change and scaling of the guided proposal  that reduces discretisation error. Numerical  examples demonstrate the performance of our methods.

\medskip

 \end{abstract}

%62M05 Markov processes: estimation
%60J60 Diﬀusion processes [See also 58J65]
%
%62F15 Bayesian inference
%65C05 Monte Carlo methods

\end{frontmatter}

\numberwithin{equation}{section}

%%%%%%%%%%%%%%%%%%%%%%5

\section{Introduction}
 %%%%%%%%%%%%%
In this article we discuss a novel approach for estimating an unknown parameter $\th \in \Theta$  of the drift and the diffusion coefficient of a diffusion process
\begin{equation}\label{sde}
 \dd X_t = b_\th(t,X_t) \dd t + \si_\th(t,X_t) \dd W_t, \qquad X_0=u  \end{equation}
which is observed discretely in time. Here $b_\th\colon \RR \times \RR^{d}$ denotes the drift function, $a_\th = \si_\th\si_\th'$ is the diffusion function, where $\si_\th\colon \RR \times \RR^d \to \RR^{d\times d'}$, and $W$ is a $d'$-dimensional Wiener process. The observation times will be denoted by $t_0 = 0 < t_1<\cdots< t_n=T$ and the corresponding observations by $x_i=X_{t_i}$. 

Estimation of $\th$ in this setting has attracted much attention during the past decade. Here we restrict attention to estimation within the Bayesian paradigm. From a theoretical perspective, results on posterior consistency have been proved in \cite{vdMeulenvZanten} and \cite{GugushviliSpreij}. The associated computational problem is the object of study here. Two  review articles that include many references on this topic are \cite{vZanten} and \cite{Sorensen}.

The main difficulty in estimation for discretely observed diffusion processes is the lack of a closed form expression for  transition densities, making  the  likelihood intractable. If the diffusion path is observed continuously, then estimation becomes easier as for a fully observed diffusion path the likelihood is available in closed form (and parameters appearing in the diffusion coefficient can be determined from the quadratic variation of the process).   This naturally suggests to study the computational problem within a missing data framework, treating the unobserved path segments between two succeeding observation times as missing data. This setup dates back to at least \cite{Pedersen}, who used it to obtain simulated maximum likelihood estimates for $\th$. Within the Bayesian computational problem,   the resulting Markov-Chain-Monte-Carlo algorithm is known as data-augmentation and was introduced in this context by \cite{Eraker}, \cite{ElerianChibShephard} and \cite{RobertsStramer}.  
This algorithm is a special form of the  Gibbs sampler  which iterates the following steps:
\begin{enumerate}
\item draw missing segments, conditional on $\th$ and the observed discrete time data;
\item draw from the distribution of $\th$, conditional on the ``full data''.
\end{enumerate}
Here, by ``full data'' we mean the path formed by the drawn segments joined at the observation times. 
The algorithm can be initialised by either interpolating the discrete time data or choosing an initial value for $\th$.  We now discuss tho major challenges for the outlined algorithm together with various solutions that have been proposed in the literature. 

\medskip

\noindent {\it Challenge 1:  generating  ``good'' proposals for the missing segments.} 
The problem of simulating diffusion bridges  has received a lot of attention  over the past 15 years. Vastly different techniques have been proposed, 
 including {\it (i)} single site Gibbs updating of the missing segments locally on a discrete grid (\cite{Eraker}), {\it (ii)}  independent Metropolis-Hastings steps using as a proposal a Laplace approximation to the conditional distribution obtained by Euler approximation (\cite{ElerianChibShephard}),  {\it (iii)} forward simulated processes derived from representations of the Brownian bridge in discrete time (\cite{DurhamGallant}), {\it (iv)}  coupling arguments (\cite{Bladt} and \cite{Bladt2}), {\it (v)}  a constrained sequential Monte Carlo algorithm with a resampling scheme guided by
backward pilots (\cite{LinChenMykland}),  and {\it (vi)}  exact simulation (\cite{BeskosPapaspiliopoulosRobertsFearnhead}).

\cite{DelyonHu} extended the work of \cite{DurhamGallant} to a continuous time setup and derived an  innovative proposal process taking the drift of the target diffusion into account. In case the diffusion coefficient is constant this proposal was proposed earlier in  \cite{Clark}.
The basic idea consists of superimposing an additional term to the drift of the unconditioned diffusion to guide the process towards the endpoint. Such proposals are termed {\it guided proposals} and have the advantage that only forward simulation of an SDE is required. More precisely, the drift of the proposal that hits $v\in \RR^d$ at time $T$ equals $b^\circ(t,x) = \la b(t,x) + (v-x)/(T-t)$, where either $\la=0$ or $\la=1$ is chosen. If $\la=0$, the guiding term $(v-x)/(T-t)$ matches with the drift of the SDE for a Brownian bridge, which indeed has drift $0$. However,  this  proposal has the drawback that it is independent of the drift $b$ of the diffusion. If $\la=1$, the guiding term  depends on $b$ and consequently there is a potential mismatch between the drift and guiding term. In both cases (i.e.\ $\la=0$ and $\la=1$) there can be a substantial mismatch between the proposals and true bridges, rendering low acceptance rates in an MH-sampler.

In \cite{Schauer} a general class of proposal processes for simulating diffusion bridges was introduced. The proposals in \cite{Schauer} do take the drift of the target diffusion into account, but in a way different from \cite{DelyonHu}. As a result, these proposals can substantially reduce the mismatch of drift and guiding term, because they allow for more flexibility in choosing an appropriate guiding term to pull the process towards the endpoint in the right manner.  An example of the advantage of this approach  is given in the introduction of \cite{Schauer}. General guidelines to exploit the added flexibility are addressed in this paper. 

For implementation purposes, any proposal has to be evaluated on a finite number of grid points. As the pulling term added to the drift for guided proposals has a singularity near the endpoint, special care is needed in choosing a discretisation method. More importantly, integrals that appear in the acceptance probability of bridges potentially suffer from this  problem as well. In this paper we introduce a time change and scaling of the proposal process that deals with these problems. 

\medskip

\noindent  {\it Challenge 2:    handling unknown parameters appearing in the diffusion coefficient.}  As pointed out by \cite{RobertsStramer}, the data augmentation algorithm degenerates if $\th$ appears in the diffusion coefficient as the quadratic variation of the full data $\int_0^T a_\th(t,X_t) \dd t$  forces  the conditional distribution for the next iterate for $\th$ to be degenerate at the current value.  Hence,   iterates of $\th$ remain stuck at their initial value. 
The problem was solved in a discretised setting by both \cite{ChibPittShephard} and \cite{GolightlyWilkinsonChapter}. Rather than updating $\th$ conditional on the discretised diffusion bridge, they proposed  updating $\th$ conditional on the increments of the Brownian motion driving the discretised diffusion bridge.  This decouples the tight dependence between $\th$ and the diffusion bridge. However, as \cite{StramerBognar} point out ``While the promising GW approach can be applied to a large class of diffusions, it is not yet rigorously justified in the literature.'' Put differently, whereas the GW (=Golightly-Wilkinson) approach works in the discretised setup, it gives no guarantee that it also works in the limit where the discretisation level tends to zero. 

 In the continuous-time framework  a solution to the aforementioned problem was given in  \cite{RobertsStramer}  for one-dimensional diffusions. It was extended to reducible multivariate diffusions (diffusions that can be transformed to have unit diffusion coefficient) in \cite{BeskosPapaspiliopoulosRobertsFearnhead} and \cite{Sermaidis}. The basic idea is that  the laws of the bridge proposals can be understood as parametrised push forwards of the law of an underlying random process common to all models with different parameters $\th$. This is naturally the case for proposals defined as solutions of stochastic differential equations and the driving Brownian motion can be taken as such underlying random process.  If $X^\star$ denotes a missing segment given that the parameter equals $\th$, the main idea consists of finding a map $g$ and a process $Z$ such that $X^\star=g(\th, Z)$.  In a more general set-up, decouplings of similar forms are discussed under the keyword \emph{non-centred parameterisation} (\cite{PapasRobertsSkoeld}).  The process $Z$ will be called the ``innovation process'' (analogous to terminology used in  \cite{ChibPittShephard} and \cite{GolightlyWilkinsonChapter}).  
Whereas in case $\si_\th=\th$ the construction is rather easy, in general proving existence  of the map $g$ and process $Z$  is subtle and this forms an important topic of this paper. We postpone a detailed discussion to Sections \ref{sec:toyproblem}  and \ref{sec:proposed-algorithm}. 

 A first attempt of finding a non-centred parameterisation in continuous time in a  general setting   was undertaken in \cite{Fuchs} (in particular section 7.4).
 \cite{Fuchs} works in the setting of \cite{DelyonHu}, so it is assumed that the diffusion coefficient $\si$ is invertible and the diffusion is time-homogeneous. While the results in \cite{Fuchs} are formulated in continuous time, the derivation involves heuristic arguments via the Lebesgue densities of the finite dimensional distributions. 
 A recent work is \cite{PapaRobertsStramer}. In their approach the missing data is initially considered in continuous time using \cite{DelyonHu} bridge proposals, but the degeneracy problem is tackled only {\it after} discretisation. 
 
What is the essential structure behind those different approaches and how can the underlying transformations be handled in continuous time without resorting to discretisation? Are these techniques tied to certain proposals, for example the proposal processes in \cite{DelyonHu},
 or are they valid for other proposal processes as well? And can conditions such as invertibility of $\sigma$ be relaxed and is it essential that the diffusion is time-homogeneous? Part of this paper consists of answering these questions in a rigorous way. As a result of this, in our setting it is evident how to replace the independence sampler for $Z$ (which updates the diffusion bridges) by  a random-walk type update on the process $Z$ in a straightforward way.

\subsection{Contribution} In this article we present a general framework for Bayesian estimation of discretely observed diffusion processes that  satisfactory deals with both aforementioned challenges. Our approach reveals the conditions necessary for obtaining an irreducible Markov chain that samples from the posterior  (after burnin). We show that the algorithm does not suffer from the degeneracy problem in case unknown parameters appear in the diffusion coefficient, not even in the continuous time setup.  The procedure can be seen as extension and unification of previous approaches within a {\it continuous time} framework.
For example the results of the rather complicated heuristics in Section 7.4 of \cite{Fuchs} appear as a special case of our work.  Specific features of our approach include: 
\begin{itemize}
\item We use in each data augmentation step ``adapted'' bridge proposals which take both the drift {\it and the value of $\th$ at that particular iteration} into account. Hence, at each iteration, the pulling term depends on $\th$, a feature which is unavailable using proposals as in \cite{DelyonHu}. Especially in the multivariate case, the additional freedom in devising good proposals is crucial for obtaining a feasible MCMC procedure. The possibility to exploit special features of the drift function to achieve high acceptance rates makes this approach interesting for practitioners. This is illustrated with a practical example in Section \ref{sec:fitz}. 

\item We provide specialised algorithms in case the drift is of the form $b_\th(\cdot)=\sum_{i=1}^N \th_i \phi_i(\cdot)$ for known functions $\phi_1,\ldots, \phi_N$ (Cf.\ algorithms \ref{alg2} and \ref{alg3}).

\item The innovation process is defined using the proposal process. As a result, in our algorithm (Cf.\ algorithm \ref{alg1}),   the innovations actually never need to be computed. This implies that our method can also cope with the case where $\si$ is not a square matrix (which is not the case for example in \cite{Fuchs}).  

\item We illustrate our work using  linear guided proposals as introduced in \cite{Schauer} and the proposals introduced in \cite{DelyonHu}.  In section \ref{sec:choiceguided} we give general guidelines on the construction of these proposals. In section \ref{sec:fitz} we show that not taking into account the drift of the diffusion can lead to extremely small acceptance probabilities for bridges.

\item 
Though we derive all our results in a continuous time setup, for implementation purposes integrals in likelihood ratios and solutions to stochastic differential equations  need to be approximated on a finite grid. 
As the drift of our proposal bridges  has a singularity near its endpoint, we introduce a time change and scaling that allows for numerically accurate discretisation and evaluation of the likelihood.
\end{itemize}

The approach with linear guided proposals can be extended to the case of partially observed diffusions, where for example some components of the diffusion are unobserved. Though the problem becomes much harder, the underlying structure for constructing an algorithm is the same. For details we refer to \cite{vdm-schauer}.

\subsection{Outline}
In Section \ref{sec:toyproblem} we clarify the aforementioned difficulties in a toy example. Here, we set some general notation and  introduce some key ideas used throughout. In  section \ref{sec:proposed-algorithm}  we precisely state our algorithms and introduce the concept of a feasible proposal. 
 In Section \ref{sec:proposals} we show that both the proposals from \cite{DelyonHu} and \cite{Schauer} are feasible.   In Section \ref{sec:choiceguided} we give guidelines on constructing a guiding term  for the proposals from \cite{Schauer}. Numerical discretisation issues and  the computational complexity of the proposed algorithms are discussed in sections \ref{sec:tc} and  \ref{sec:impl} respectively. 
 Numerical examples are given in  Section \ref{sec:examples}. The appendix contains  a few postponed proofs.

\section{A toy problem}\label{sec:toyproblem} 
In this section we consider a toy example to illustrate some key ideas to solve the aforementioned problems with a simple data-augmentation algorithm. The type of reparameterisation introduced shortly is not new, and has appeared for example in \cite{RobertsStramer}.  The goal here is to introduce key ideas and point out some of its potential shortcomings in more complex problems. Furthermore, later on we will deal with more difficult cases and this toy example  allows us to sequentially build up an appropriate framework for that.
We consider the diffusion process  
\[\dd X_t = b(X_t) \dd t + \th\dd W_t, \quad X_0=u,\quad t\in[0,T], \] where $b$ is a known drift function and $\th \in \Theta$ an unknown scaling parameter. We assume $\th$ is equipped with a prior distribution $\pi_0(\th)$ and only one observation $X_T = v$ at time $T$  is available. We aim to draw from  the posterior $\pi(\th\mid X_T)$.
 The diffusion process conditioned on $X_T = v$ is a diffusion process itself. Denote by $X^\star$ the conditioned diffusion path $(X_t,\, t\in (0,T))$ (conditional on $X_T=v$). Suppose we wish to iterate a data-augmentation algorithm and the current iterate is given by $(X^\star, \th)$.

{\it Updating $X^\star$:}  
  For almost all choices of $b$, there is no direct way of simulating $X^\star$. Instead, one can first generate a proposal bridge $X^\circ$ and accept with MH-acceptance probability. As an easy tractable example we choose to take  
\begin{equation}\label{eq:Xo-example}	
	\dd X^\circ_t= \frac{v-X^\circ_t}{T-t} \dd t +  +\th \dd W_t, \quad X^\circ_0=u
\end{equation}
 where $W$ is a Brownian Motion on $[0,T]$.

Denote the laws of $X^\circ$ and $X^\star$ (viewed as Borel measures on $C([0,T], \RR^d)$) by $\bb{P}^\circ$ and $\bb{P}^\star$ respectively.  We have
\begin{equation}\label{eq:llratio}	\frac{\dd \bb{P}^\star_\th}{\dd \bb{P}^\circ_\th}(X^\circ) = \frac{\tilde{p}_\th(0,u; T,v)}{p_\th(0,u; T,v)} \Psi_\th(X^\circ), 
\end{equation}
with 
\[ \Psi_{\th}(X^\circ) = \exp\left( \th^{-2}\int_0^T b(X^\circ_s) \dd X^\circ_s -\frac12\th^{-2} \int_0^T b(X^\circ_s)^2 \dd s\right). \]
Here, $p$  denotes the transition densities of the process $X$ and $\tilde{p}(0,u;T,v)=\phi(v; u, \th^2 T)$ (the  density of the $N(u,\th^2 T)$-distribution, evaluated at $v$).  Absolute continuity is a consequence of Girsanov's theorem applied to the unconditioned processes and the abstract Bayes' formula. Now the MH-step consists of generating a proposal $X^\circ$ and accepting it with probability $1 \wedge \left(\Psi_\th(X^\circ)/ \Psi_\th(X^\star)\right)$ (the ratio of transition densities just acts as a proportionality constant here).

{\it Updating $\th$:}  
As explained in the introduction, taking the missing segment as missing data yields the Metropolis-Hastings algorithm reducible. To deal with this problem, note that by equation \eqref{eq:Xo-example}, there exists a mapping $g$ such that  $X^\circ=g(\th, W)$.
Define the process $Z$ by the relation
\begin{equation}\label{eq:defZstar} X^\star = g(\th, Z). \end{equation}

Now that $Z$ is defined, rather then drawing from the distribution of $\th$ conditional on $(X_0=u, X_T=v, X^\star)$ we will sample from the the distribution of $\th$ conditional on $(X_0=u, X_T=v, Z)$. This means that we augment the discrete time observations with $Z$ instead of $X^\star$. 
Denote the laws of $Z$ and $W$ by $\bb{Z}_\th$ and $\bb{W}$
respectively. Suppose the current iterate is $(\th, Z)$, where $Z$ can be extracted from $\th$ and $X^\star$ by means of equation \eqref{eq:defZstar}. 
The following diagram summarises the notation introduced:
\begin{equation}
\label{eq:measures}
\centering
\begin{tabular}{ l || c | c} 
Process& $Z \stackrel{g(\th, \cdot)}{\longrightarrow} X^\star$ & $W \stackrel{g(\th, \cdot)}{\longrightarrow} X^\circ$ \\
\hline 
Measure & $\mathbb{Z}_\th \qquad \mathbb{P}^\star_\th$ & $\mathbb{W} \qquad  \mathbb{P}^\circ_\th$ \\
\end{tabular}
\end{equation}

For updating $\th$ we propose a value ${\th^\circ}$ from some proposal distribution $q(\cdot \mid \th)$ and accept the proposal with probability $\min(1,A)$, where 
\begin{equation}
\label{eq:AA}
 A = \frac{\pi_0({\th^\circ})}{\pi_0(\theta)}\frac{p_{\th^\circ}(0,u; T,v)}{p_\theta(0,u; T,v)} \frac{\dd \bb{Z}_{\th^\circ}}{\dd \bb{Z}_\theta}(Z) \frac{q(\th \mid {\th^\circ})}{q({\th^\circ} \mid \theta)}. 
\end{equation}
Here, we have implicitly assumed that $\bb{Z}_{\th^\circ}$ and $\bb{Z}_\th$ are equivalent, which is indeed the case as  we have
%\begin{align*}
\[	\frac{\dd \bb{Z}_{\th^\circ}}{\dd \bb{Z}_\th}(Z) =\frac{\dd \bb{Z}_{\th^\circ}}{\dd \bb{W}}(Z) \bigg/ \frac{\dd \bb{Z}_\th}{\dd \bb{W}}(Z)%\\ & 
	= 
\frac{\dd \bb{P}^\star_{\th^\circ}}{\dd \bb{P}^\circ_{\th^\circ}} (g({\th^\circ}, Z)) \bigg/
\frac{\dd \bb{P}^\star_\theta}{\dd \bb{P}^\circ_\theta} (g(\theta, Z))
\]
%\end{align*}
and thus results from absolute continuity of $\bb{P}^\star_\th$ and $\bb{P}^\circ_\th$. 
By equation \eqref{eq:llratio}, we now get 
\[	\frac{\dd \bb{Z}_{\th^\circ}}{\dd \bb{Z}_\th}(Z)= \frac{p_\th(0,u; T,v)}{p_{\th^\circ}(0,u; T,v)} \frac{\tilde{p}_{\th^\circ}(0,u; T,v)}{\tilde{p}_\theta(0,u; T,v)}
 \frac{\Psi_{\th^\circ}(g({\th^\circ}, Z))}{\Psi_{ \theta}(g(\theta, Z))}. 
\]
Substituting this expression into equation \eqref{eq:AA} yields
\begin{equation} \label{eq:Atoy}
 A=  \frac{\pi_0({\th^\circ})}{\pi_0(\theta)} \frac{\tilde{p}_{\th^\circ}(0,u; T,v)}{\tilde{p}_\theta(0,u; T,v)}
 \frac{\Psi_{{\th^\circ}}(g({\th^\circ}, Z))}{\Psi_\th(g(\theta, Z))}  \frac{q(\th \mid {\th^\circ})}{q({\th^\circ} \mid \theta)}
\end{equation}
and all terms containing the unknown transition density cancel.  In Section \ref{sec:proposals} we will show that cancellation of $p$ from the acceptance probability holds much more generally. 

The feasibility and efficiency of this algorithm is crucially determined by  choice of the transition kernel $q$ and proposal process $X^\circ$. We focus  on an appropriate choice of $X^\circ$, though in section \ref{subsec:conj} we give guidelines on appropriate choice of $q$ if the drift possesses a specific structure with respect to $\theta$.

 %%%%%%%%%%%%%
\section{Proposed MCMC algorithms}\label{sec:proposed-algorithm}

Our starting point is that under weak assumptions the target diffusion bridge $X^\star$ from $u$ at time $t=0$ to $v$ at time $t=T$ is characterised as the solution to the SDE
\begin{equation}\label{xstar} 
   \dd X^\star_t=   b_\th^\star(t,X^\star_t) \dd t + \si_\th(t, X^\star_t) \dd W_t,\qquad X^\star_0=u,\qquad t \in [0,T),
\end{equation}
where 
\begin{equation}\label{bstar} 
b^\star_\th(t,x) =b_\th(t,x)+ a_\th(t,x) \nabla_x \log p_\th(t,x;T,v)
\end{equation}
and $a = \sigma\sigma^\T$. Here the transition density of $X$ is denoted by $p_\th$ and $p_\th(t, x; T,v)$  is the density of the process starting in $x$ at time $t$, ending in $v$ at time $T$.
%---
\subsection{Innovation process}
  Direct forward simulation of $X^\star$  is hardly ever possible, as $p$ is intractable. Instead, we propose to simulate a process $X^\circ$ with induced law that is absolutely continuous with respect to that of $X^\star$. More precisely, we assume the {\it proposal process} $X^\circ$ satisfies the SDE
\begin{equation}\label{xcirc} 
\dd X^\circ_t = b_\th^\circ(t, X^\circ_t) \dd t + \sigma_\th(t, X^\circ_t) \dd W_t, \qquad X^\circ_0=u,\qquad t \in [0,T)
\end{equation}

We now describe a general parametrisation to decouple the dependence between the latent paths of the diffusion between discrete time observations and the parameter $\th$.

%%%%%%%%
The following proposition is key to the definition of the map $g$. In its statement we refer to the canonical setup on which an exact SDE can be solved, details are in section V.10 of \cite{RogersWilliams2}. 

\begin{prop} 
\label{prop:existence-g}
Assume 
\begin{itemize}
\item the SDEs  for $X^\circ$ and $X^\star$ are pathwise exact (in the sense of definition V-9.4 of \cite{RogersWilliams2});
\item there exists a strong solution for the SDE for $X^\circ$ (in the sense of definition V.10.9 of \cite{RogersWilliams2}) jointly measurable with respect to starting point, parameter and path $W$;
\item $\bb{P}^\circ$ and $\bb{P}^\star$ are absolutely continuous.
\end{itemize} 
Then there exists a map $g$ and a Wiener process $W$  such that $X^\circ=g(\th, W)$ on the canonical setup. Furthermore, there exists a process $Z$ such that $X^\star=g(\th, Z)$. The process $Z$ satisfies the SDE
\begin{equation}\label{eq:defZstar2} \dd Z_t = \mu_\th\big(Z_t\big) \dd t + \dd W_t \end{equation}
where the map $\mu_\th$ satisfies 
\begin{equation}\label{eq:def-mu} \si_\th(t,x)  \mu_\th(t,x) =b_\th^\star(t,x)-b_\th^\circ(t,x) .\end{equation}
Moreover,
\begin{equation}
\label{eq:relation-Zstar-Xstar} \dd X^\star_t = b^\circ_\th(t, X^\star_t) \dd t + \sigma_\th(t, X^\star_t) \dd Z_t, \qquad X^\star_0=u,\qquad t \in [0,T).
\end{equation}

\end{prop}
\begin{proof}Denote the law of $W$ by $\bb{W}$.
Existence of $g$ such that $X^\circ=g(\th, W)$ is implied by existence of a strong solution for the SDE for $X^\circ$. 
If $Y$ satisfies \[\dd Y_t = b_\th^\circ(t, Y_t) \dd t + \sigma_\th(t, Y_t) \dd W_t, \qquad Y_0=u,\qquad t \in [0,T) \]
then 
$Y = g(\th,W)$. 
Define
\[ L_\th = \exp\left(\int_0^T \mu_\th(t,Y_t) \dd W_t - \int_0^T \mu_\th(t,Y_t)^2 \dd t\right) \]
and assume for the moment that  $E^{\mathbb{W}} L_\th=1$. Define the measure $\mathbb{Z}_\th$ by $\dd \mathbb{Z}_\th = L_\th \dd \mathbb{W}$. 
By Girsanov's theorem, it follows that the process $Z$ defined by equation \eqref{eq:defZstar2} is a Brownian Motion under the measure $\mathbb{Z}_\th$. If we define $Y$ by $Y=g(\th,Z)$  then
\[\dd Y_t = b_\th^\circ(t, Y_t) \dd t + \sigma_\th(t, Y_t) \dd Z_t, \qquad Y_0=u,\qquad t \in [0,T) \]
under $\mathbb{Z}_\th$.
Plugging \eqref{eq:defZstar2}  into this equation shows that 
\[\dd Y_t = b^\star_\th(t, Y_t) \dd t + \sigma_\th(t, Y_t) \dd W_t, \qquad Y_0=u,\qquad t \in [0,T). \]
By pathwise uniqueness $Y=X^\star$ up to indistinguishability under $\bb{W}$ (because $W$ is a Wiener process). 
Hence, $X^\star=g(\th,Z)$ and \eqref{eq:relation-Zstar-Xstar} follows.
 We have
\begin{equation}\label{change}
 \frac{\dd \mathbb{Z}_\th}{\dd \mathbb{W}}(\cdot) =  \frac{\dd \mathbb{P}_\th^\star}{\dd \mathbb{P}_\th^\circ} (g(\th,\cdot))
\end{equation}
and henceforth  existence of $\mu_\th$ such that $E^{\mathbb{W}} L_\th=1$ follows from our assumption that $\mathbb{P}^\star_\th$ and $\mathbb{P}^\circ_\th$ are absolutely continuous. 
\end{proof}

We refer to the process $Z$ as the  {\it innovation process} corresponding to $X^\star$ (by analogy of the terminology of \cite{GolightlyWilkinsonChapter} and  \cite{ChibPittShephard}). 
Clearly, $X^\star$ is related to $Z$ just like $X^\circ$ is related to $W$. Note however that while the law of  $W$ does not depend on $\th$ under $\mathbb{W}$, the law of $Z$ does depend on $\th$ under $\mathbb{W}$.

In the following we will denote the Radon-Nikodym derivative between $\mathbb{P}^\circ$ and $\mathbb{P}^\star$ by $\Phi$: \[ \Phi_\th:= \frac{\dd \mathbb{P}_\th^\star}{\dd \mathbb{P}_\th^\circ}.\]

%%%%%%%%%%%%%%

\subsection{Algorithm}\label{sec:algo}
In this section we present an algorithm to  sample from the posterior of $\th$  given the discrete observations  $\scr{D}=\{X_0 = u, X_{t_1} = x_1, \ldots, X_{t_n} = x_{n}\}$. Denote the prior density on $\th$ by $\pi_0$ and let $q(\th^\circ \mid \th)$ be the density for proposing $\th^\circ$ given the current value $\th$. 
The idea is to define a Metropolis--Hastings sampler on $(\th, Z)$ instead of $(\th, X^\star)$ where $Z$ is the innovation process from the previous section.  

More precisely, we construct  a Markov chain for  $(\th, (Z_i)_{ 1\le i \le n})$, where each $Z_i$ is an innovation process corresponding to the bridge $X^\star_i$ connecting observation $x_{i-1}$ to $x_{i}$.

\begin{algorithm}\label{alg1}\

\begin{enumerate}
\item {\bf Initialisation.} Choose a starting value for $\th$ and sample $i = 1,\dots, n$ Wiener processes $W_i$ and set $Z_i = W_i$. 

\item {\bf Update $Z \mid (\th, \scr{D})$.} Independently, for $1\le i \le n$ do
\begin{enumerate}
\item  Sample a Wiener process $Z^\circ_i$.
\item  Sample $U \sim \scr{U}(0,1)$. Compute
\[ A_1 =\frac{\Phi_{\th}(g(\th,Z^\circ_i))}{\Phi_{\th}(g(\th, Z_i))}. \]
Set 
\[ Z_i := \begin{cases} Z^\circ_i& \text{if}\quad U\le A_1 \\
 Z_i & \text{if}\quad U> A_1 \end{cases}. \]
\end{enumerate}

\item {\bf Update $\th \mid (Z, \scr{D})$.} 
\begin{enumerate}
\item  Sample ${\th^\circ} \sim q(\cdot \mid \th)$.
\item  Sample $U \sim \scr{U}(0,1)$. Compute
\[	 A_2=\frac{q(\th \mid {\th^\circ})}{q({\th^\circ} \mid \th)}  \frac{\pi_0({\th^\circ})}{\pi_0(\th)} \prod_{i=1}^n \frac{p_{{\th^\circ}}(t_{i-1},x_{i-1}; t_i, x_i)}{p_{\th}(t_{i-1},x_{i-1}; t_i, x_i)}
\frac{\Phi_{\th^\circ}(g({\th^\circ}, Z_i))}{\Phi_{\th}(g(\th,Z_i))} \]
Set 
\[\th := \begin{cases} {\th^\circ} & \text{if}\quad U\le A_2 \\
 \th & \text{if}\quad U> A_2 \end{cases}. \]
\end{enumerate}
\item Repeat steps (2) and (3).
\end{enumerate}
\end{algorithm}

Note that in none of these steps we need to compute innovations $Z$ from $X^\star$.  This is a consequence of adapting the definition of the innovations to the bridge proposals being used. 

In step (2) an independent Metropolis-Hastings step is used. Instead, one can also propose $Z^\circ$ based on the current value of $Z$ in the following way
\begin{equation}\label{eq:rw-innovation} Z^\circ_t = \sqrt{\rho} Z_t + \sqrt{1-\rho} W_t, \end{equation}
where $\rho\in [0,1)$ and $W$ is a Wiener process  under $\bb{W}$ that is independent of $Z$.
In this case 
\[ A_1 = \frac{\left(\!\dd \bb{Z}_\th \big/\!\dd\bb{W}\right) (Z)}{\left(\!\dd \bb{Z}_\th \big/\! \dd \bb{W}\right) (Z^\circ)}
\frac{\dd Q_\rho}{\dd Q^T_\rho}(Z^\circ, Z), \]
where $Q_\rho(x,y)=Q_\rho^T(y,x)$. Here we use the general formulation of the Metropolis-Hastings algorithm as explained in \cite{tierneyAnnapl}.  The second term equals one by symmetry of $Q(\cdot, \cdot)$. 
This implies that the acceptance probability in step 2(b) remains the same. 

\begin{rem}
Different proposals can be obtained by varying $b^\circ$ in \eqref{xcirc} and it   is clear that  the mapping $g$ varies accordingly. A good choice obviously affects the acceptance probability of step 2 in algorithm \ref{alg1}. However, it affects the acceptance probability of step 3 as well as this step is a joint update of $(\th, X^\star)$. This implies that a proposal $\th^\circ$ in step 3 which is ``good'' (in the sense of being like a draw from the posterior of $\th$), may nevertheless be rejected if the mapping $g$ is such that  $g(\th^\circ, Z)$ does not resemble a bridge with drift and diffusion coefficient indexed by $\th^\circ$.   Ideally, one would take $g=g_{\rm opt}$, where $g_{\rm opt}$ is defined by the relation $X^\star = g_{\text{\rm opt}}(\th, W)$, with $W$ denoting a Wiener process.
\end{rem}

\begin{thm}\label{lem:irreducible}
Suppose $q(\th^\circ \mid \th)$ is almost everywhere strictly positive on the support  of the prior for $\th$. Then the chain induced by  algorithm \ref{alg1} is irreducible.
\end{thm}

\begin{proof}
Step 2 constitutes a step of a MH-sampler with independent proposals. The expression for $A_1$ follows directly from equation \eqref{change}. The expression for $A_2$  in step 3 follows in exactly the same way as equation \eqref{eq:Atoy} was established in the toy-example (Cf.\ section \ref{sec:toyproblem}). 
 The remaining observation needed is the following: 
As $\Phi_\th$ is the Radon-Nikodym derivative between two equivalent distributions, it is almost surely strictly positive and finite. Since the transition densities are strictly positive as well,  both $A_1$ and $A_2$ are strictly positive and the result follows.
\end{proof}

At first sight, it may seem that algorithm \ref{alg1} is not of much practical value. First of all, the mapping $g$ is unknown. However, as any algorithm derived in continuous time ultimately has to be approximated by discretisation, we can choose a discretisation level and compute $Y = g(\th, Z)$ on a fine grid by discretising the stochastic differential equation
\[ \dd Y_t = b^\circ_\th(t, Y_t) \dd t + \si_\th(t, Y_t) \dd Z_t. \]
Second, it seems impossible to compute the acceptance probabilities in steps 2 and 3 because $\Phi_\th$ depends on $p$ and $p$ explicitly pops up in the formula for $A_2$. However, it turns out that for many choices of $b^\circ$ the unknown transition density $p$ only appears as a multiplicative constant in $\Phi_\th$ such that it cancels the $p$ in the expression for $A_2$. For future reference, we introduce the following definition. 
\begin{defn}\label{def:feasible}
We call a proposal $X^\circ$ as defined in equation \eqref{xcirc} {\it feasible} if $b^\circ_\th$ is such that both $A_1$ and $A_2$ appearing in algorithm \ref{alg1} do not depend on the transition density $p$.
\end{defn}

In section \ref{sec:proposals} we will give examples of classes of feasible proposals.

 %%%%%%%%%%%%%%%%%%%%%%%%%
\subsection{Partially conjugate series prior for the drift}\label{subsec:conj}
In this subsection we study specific cases of algorithm \ref{alg1}  when the drift is of the form
\begin{equation}\label{sumprior}
b_{\vartheta}(x) = \sum_{i=1}^N \vartheta_i \phi_i(x)
\end{equation}
where $\vartheta=(\vartheta_1, \ldots,\vartheta_N)$ is an unknown  parameter and $\phi_1,\ldots, \phi_N$ are known  functions on $\RR^d$. We assume the diffusion coefficient  is parametrised by the parameter $\ga$. We denote the vector of all unknown parameters by $\th=(\vartheta, \gamma)$ and assume these are assigned independent priors. With slight abuse of notation we use $\pi_0(\vartheta)$ and $\pi_0(\gamma)$ to denote the priors on   $\vartheta$ and $\gamma$ respectively (the argument in parentheses will clarify which prior is meant). In this case it is convenient to choose a conjugate Gaussian prior for the coefficients, $\vartheta_i \sim \scr N(0, \xi^2_i)$, 
for positive scaling constants $\xi_i$. Priors for the drift obtained by specifying a prior distribution on $\vartheta$ were previously considered in \cite{Kuchler-Sorensen}, \cite{Bladt} and \cite{vdMeulen}. Upon completing the square, it follows that  the distribution of $\vartheta$ conditional on $\gamma$ and the full path $Y$ of the diffusion  is multivariate normal with mean vector $W_\ga^{-1} \mu_\ga$ and covariance matrix $W_\ga^{-1}$. 
We define for $k, \ell \in \{1,\ldots, d\}$,
\begin{align*} \mu_\ga[k] &= \int_0^T \phi_k(Y_t)'a_\gamma^{-1}(Y_t)\dd Y_t \\ 
\Sigma_{\ga}[k,\ell] & = \int_0^T \phi_k(Y_t)' a_\gamma^{-1}(Y_t) \phi_\ell(Y_t) \dd t \\
W_\ga &=\Sigma + \diag(\xi_1^{-2},\ldots, \xi_N^{-2}).
\end{align*}
(For a vector $x\in \RR^n$ we denote the $i$-th element by $x[i]$. To emphasise the dependence on $Y$ we sometimes also write $\mu_\gamma(Y)$, $W_\gamma(Y)$ etc).
This leads to a natural adaptation of algorithm \ref{alg1} from section \ref{sec:algo}. 
\begin{algorithm}\label{alg2}
Steps 1, 2 and 4 as in algorithm \ref{alg1}. {\it Assume that $\sigma$ is invertible.}  Step 3 is given by 
\begin{enumerate}
\item[3.1] {\bf Update }$\gamma  \mid (\vartheta, Z, \scr{D})$.

\begin{enumerate}
\item Sample ${\gamma^\circ} \sim q(\cdot \mid \gamma)$.
\item  Sample $U \sim \scr{U}(0,1)$. Compute
\[	 A_3=\frac{q(\gamma \mid {\gamma^\circ})}{q({\gamma^\circ} \mid \ga)}  \frac{\pi_0({\gamma^\circ})}{\pi_0(\ga)}\prod_{i=1}^n\frac{p_{({\gamma^\circ},\vartheta)}(t_{i-1},x_{i-1}; t_i, x_i)}{p_{(\ga,\vartheta)}(t_{i-1},x_{i-1}; t_i, x_i)}
\frac{\Phi_{ ({\gamma^\circ},\vartheta)}(g(({\gamma^\circ},\vartheta), Z_i)}{\Phi_{ (\ga,\vartheta)}(g((\ga,\vartheta),Z_i))}  \]
Set 
\[\ga := \begin{cases} {\gamma^\circ} & \text{if}\quad U\le A_3 \\
 \ga & \text{if}\quad U> A_3 \end{cases}. \]
\end{enumerate}

\item[3.2] {\bf Update }$\vartheta \mid  (\gamma, Z,  \scr{D})$.  
\begin{enumerate}
\item Compute $\mu_g = \mu_\ga(g((\vartheta, \gamma), Z))$ and $W_\gamma = W_\ga(g((\vartheta, \gamma), Z))$.
\item Sample ${\vartheta^\circ} \sim \scr N(W_\ga^{-1}\mu_\ga, W_\ga^{-1})$.
\item Compute $Z^\circ$ such that $g(({\vartheta^\circ}, \gamma), Z^\circ) = g((\vartheta, \gamma), Z)$. Set  $\vartheta = \vartheta^\circ$ and $Z = Z^\circ$.
\end{enumerate}

\end{enumerate}
\end{algorithm}
Note that computation of $Z^\circ$ in step 3.2(c) requires invertibility of $\si$. 

\begin{proof}
Suppose $(\vartheta, \gamma, Z) \sim \pi $, where $\pi$ denotes  the posterior distribution. Consider the map $f\colon (\vartheta, \gamma, Z) \mapsto (\vartheta, \gamma,  X^\star)$, where $X^\star=g((\vartheta, \gamma), Z)$. 
We show that step 3.2 preserves $\pi$. The distribution of $(\vartheta, \gamma, X^\star)$ is the image measure of the posterior distribution $\pi$ of the tuple $(\vartheta, \gamma, Z)$ under $f$ and coincides with the posterior distribution of $(\vartheta, \gamma, X^\star)$. Denote the image measure of $\pi$ under $f$ by  by $\pi \circ f^{-1}$.   In steps 3.2(a) and 3.2(b) we apply the mapping $f$, followed by a Gibbs step  in which we draw $\vartheta^\circ$ conditional on $(\gamma, X^\star)$. The latter preserves $\pi \circ f^{-1}$. Hence $(\vartheta^\circ, \gamma, X^\star) \sim \pi \circ f^{-1}$. In step 3.2(c) we we compute $(\vartheta^\circ, \gamma, Z^\circ)$ as pre-image of $(\vartheta^\circ, \gamma, X^\star)$ under $f$ (this is possible as we assume $\si$ to be invertible). Hence $(\vartheta^\circ, \gamma, Z^\circ) \sim \pi$.
\end{proof}

A variation of this algorithm is obtained in case the drift is of the form specified in equation \eqref{sumprior} and the diffusion coefficient depends on both $\vartheta$ and $\ga$. In this case we can update $\gamma$ just as in algorithm \ref{alg2}. Updating $\vartheta$ can be done using a random walk type proposal of the form
\[ q(\vartheta^\circ \mid \vartheta) \sim N(\vartheta, \alpha V), \]
with $\alpha$ a positive tuning parameter. Motivated by the covariance matrix of the prior exploited in the case of partial conjugacy we propose to replace $V$ by $W^{-1}_{(\vartheta,\gamma)}$. 
By this choice, if two components $\vartheta_i$ and $\vartheta_j$ are  strongly correlated, the proposed local random walk proposals have the same correlation structure, which can improve mixing of the chain.

\begin{algorithm}\label{alg3} The same algorithm as Algorithm \ref{alg2} without the invertibility assumptions and Step 3.2 replaced by
\begin{enumerate}
\item[3.2'] {\bf Update } $\vartheta \mid (\gamma,  Z , \scr{D})$. 
\begin{enumerate}
\item Set $X^\star=  g(\vartheta, Z)$.
\item Compute $W_{(\vartheta,\ga)}$.
\item Sample ${\vartheta^\circ} \sim \scr N(\vartheta, \alpha^2 W_{(\vartheta,\ga)}^{-1})$.
\item Compute $W_{({\vartheta^\circ},\ga)}$.
\item   Sample $U \sim \scr{U}(0,1)$. Compute
\begin{align*}	 A_4 &=   \frac{|W_{{\vartheta^\circ}}|^{1/2}}{|W_{\vartheta}|^{1/2}}  \exp\left(- \frac{1}{2\alpha^2}  ({\vartheta^\circ}-\vartheta)'  (W_{\vartheta^\circ}-W_{\vartheta})  ({\vartheta^\circ}-\vartheta)\right)\\ & \qquad \times \frac{\pi_0({\vartheta^\circ})}{\pi_0(\vartheta)}\prod_{i=1}^n \frac{p_{(\ga,{\vartheta^\circ})}(t_{i-1},x_{i-1}; t_i, x_i)}{p_{(\ga,\vartheta)}(t_{i-1},x_{i-1}; t_i, x_i)}
\frac{\Phi_{(\ga,{\vartheta^\circ})}(g((\ga,{\vartheta^\circ}), Z_i))}{\Phi_{ (\ga,\vartheta)}(g((\ga,\vartheta),Z_i))} .\end{align*}
Set 
\[\vartheta := \begin{cases} {\vartheta^\circ} & \text{if}\quad U\le A_4 \\
 \vartheta & \text{if}\quad U> A_4 \end{cases}. \]
\end{enumerate}
\end{enumerate}
\end{algorithm}

The following argument gives some guidance in the choice of $\alpha$. If the target distribution is a $d$-dimensional Gaussian distribution $\scr N_d(\mu, \Sigma)$ and the proposal is of the form ${\vartheta^\circ} \sim q({\vartheta^\circ},\vartheta) \sim \scr N_d (\vartheta , \alpha^2\Sigma_q)$, then optimal choices for $\alpha$ and $\Sigma_q$ are given by  $\Sigma_q = \Sigma$ and $\alpha = 2.38/\sqrt d$, cf.\ \cite{RosenthalChapter}. Hence, we will choose $\alpha=2.38/\sqrt{\dim(\vartheta)}$, which corresponds to an average acceptance probability equal to $0.234$.  Although this procedure will not be optimal for the examples considered, it provides an automatic choice and avoids tedious pilot runs.

%%%%%%%%%%%%%%%%%%%%%%%%%%%%%%

\section{Feasible proposals}\label{sec:proposals}

In this section we discuss examples of proposals that enable application of algorithm \ref{alg1}. First we discuss the prerequisites for this in general.   Trivially, we should be able to sample a discretised version of the process $X^\circ$. This can  be done using a discretisation method for stochastic differential equations, such as Euler-discretisation. Secondly, it is required that the assumptions of proposition  \ref{prop:existence-g} are satisfied.  
Third, we need our proposal to be feasible in the sense of definition \ref{def:feasible}. This requires choosing $b^\circ$ such that $\Phi_\th = \dd \mathbb{P}_\th^\star/\dd \mathbb{P}_\th^\circ$   contains the transition density $p$ solely as a multiplicative factor in the denominator. {\it As $\th$ is fixed throughout this section, we drop it temporarily from our notation.} It is not too hard to see why $p$ would only show up  as a multiplicative factor in the denominator. Denote the laws of $X$, $X^\circ$ and $X^\star$ on $C[0,t]$ by $\PP^t$, $\PP^{\circ,t}$ and $\PP^{\star,t}$ respectively. If $t=T$ we will omit  time dependence. We have 
\[ \frac{\dd \PP^{\star,t}}{\dd \PP^{\circ,t}}(X^\circ) = \frac{p(t,X^\circ_t; T,v)}{p(0,u; T,v)} \frac{\dd \PP^{t}}{\dd \PP^{\circ,t}}(X^\circ) \]
(see for instance the proof of proposition 1 in \cite{Schauer}). Hence $p$ shows up only in the first term on the right-hand-side. Upon taking the limit $t\uparrow T$ of the expectation on the right-hand-side, the term $p(t,X^\circ_t; T,v)$ may vanish, depending on the precise form of $b^\circ$. For the proposals of sections \ref{sec:DHprop} and \ref{sec:guided-proposals} ahead, a formal proof of this can be found in \cite{DelyonHu} and \cite{Schauer}  respectively. In the following we will  sketch the argument for the disappearance of $p(t,X^\circ_t; T,v)$ under $t\uparrow T$.

%------------------------
\subsection{Proposals by Delyon and Hu}\label{sec:DHprop}
\cite{DelyonHu} introduced proposals for which 
\begin{equation}\label{xcirc2} 
b^\circ(t,x) =  \la b(t, x) + \frac{v-x}{T-t}, \end{equation}
where $\la \in \{0,1\}$. When evaluated for $x=X^\circ_t$, the pulling term $(v-X^\circ_t)/(T-t)$  forces $X^\circ$ to hit $v$ at time $T$.  Sufficient conditions for absolute continuity  and  expressions for the likelihood ratio of the laws of $X^\star$ and $X^\circ$ are derived in \cite{DelyonHu}. However,  the proportionality constants in the derived likelihood ratio are missing. Whereas for generating diffusion bridges using a MH-sampler these constants are irrelevant, they do matter for step 3 of algorithm \ref{alg1} (because the constants depend on $\th$). In case of a one-dimensional diffusion, the constant in the Radon-Nikodym derivative is derived in \cite{PapaspiliopoulosRoberts}. The extension to the multivariate case brings no surprises. Here we consider the case $\la=0$.  It turns out that  the  derivative can be obtained by rewriting the expression obtained from applying Girsanov's theorem
\begin{align}\label{eq:dhderivation}  &\frac{\dd \bb{P}^{\star,t}}{\dd \bb{P}^{\circ,t}}(X^\circ) = 
\exp\left(J_t(X^\circ)\right) \times \frac{p(t, X^\circ_t; T,v)}{\varphi(v; X^\circ_t, (T-t) a(t,X^\circ_t))} \\ &\quad \times   \frac1{p(0,u; T,v)} (2\pi T)^{-d/2} |\det a(t,X^\circ_t)|^{-1/2} \exp\left(-\frac1{2 T}   (v-u)^\T a(0,u)^{-1}(v-u)\right).\nonumber
\end{align}
  Here $\varphi(x; \mu,a)$ denotes the value of the normal density with mean $\mu$ and variance $a$, evaluated at $x$ and the functional $J_t$ is defined by 
\begin{align*} J_t(X^\circ)&= \int_0^t b(s,X^\circ_s)^\T a^{-1}(s,X^\circ_s) \dd X^\circ_s- \frac12 \int_0^t b(s,X^\circ_s)^\T  a^{-1}(s,X^\circ_s) b(s,X^\circ_s)\dd s  \\ & -\frac12 \int_0^t (T-s)^{-1}(v-X^\circ_s)^\T  \diamond \dd a^{-1}(s,X^\circ_s) (v-X^\circ_s),  \end{align*}
where  the $\diamond$-integral is obtained as the limit of sums where the integrand is computed at the right limit of each time interval as opposed to the left limit used in the definition of the It\=o integral. It can be shown that all terms are well-behaved under the limit $t\uparrow T$ and that 
\[ \Phi(X^\circ)= \exp\left(J_T(X^\circ)\right) \frac{\varphi(v; u, a(0,u))}{p(0,u; T,v)} \sqrt{\left|\frac{\det a(0,u)}{\det a(T,v)}\right|}. \] The term $p(t, X^\circ_t; T,v)$ appearing in \eqref{eq:dhderivation}  is essentially cancelled by $\varphi(v; X^\circ_t, (T-t) a(t,X^\circ_t))$ in the limit.
From the expression for $\Phi$  we see that the factor $p(0,u; T,v)$ solely appears as a multiplicative constant in the denominator of the Radon-Nikodym derivative between the target bridge and proposal bridge. Therefore, the proposals derived from \eqref{xcirc2} are feasible.

%------------------------

\subsection{Guided proposals}\label{sec:guided-proposals}
In this section we review a flexible class of proposal processes $X^\circ$ that was developed and studied in  \cite{Schauer}. We will use this framework in the remainder and provide a recap of the relevant results in this section. For precise statements of these results we refer the reader to \cite{Schauer}. 

The basic idea is to replace the generally intractable transition density $p$ that appears in the  dynamics of the target bridge (see equations \eqref{xstar} and \eqref{bstar}) by  the transition density of a diffusion process $\tilde{X}$  for which it is known in closed form. Assume $\tilde{X}$ satisfies the SDE $\dd \tilde{X}_t = \tilde{b}(t, \tilde{X}_t) \dd t + \tilde{\si}(t, \tilde{X}_t) \dd W_t$.  Denote the transition density of $\tilde{X}$ by $\tilde{p}(s,x; T,v)$ and set $\tilde{a}=\tilde{\si}\tilde{\si}^\T$. Define the process $X^\circ$ as the solution of the SDE \eqref{xcirc} with 
\begin{equation}\label{bcirc} \tag{$\circ\circ$}
			b^\circ(t, x) = b(t,x) + a(t, x) \nabla_x  \log  \tilde p(t,x;T,v).
\end{equation} 
A process $X^\circ$ constructed in this way is referred to as a {\it guided proposal} (a guiding term is superimposed on the drift to ensure the process hits $v$ at time $T$). 

We reduce notation by writing $p(s,x)$ for $p(s,x; T,v)$. 
Define 
\begin{equation}\label{eq:notation}
 R(s,x) = \log  p(s,x), \quad
	 r(s,x) = \nabla  R(s,x), \quad  H(s,x) = -\Delta  R(s,x), 
\end{equation}
where $\nabla$ and $\Delta$ denote the gradient and Laplacian with respect to $x$ respectively. Similarly, write $\tilde{p}(s,x)$ instead of $\tilde{p}(s,x; T,v)$, etc. In  \cite{Schauer} sufficient conditions for absolute continuity of $\PP^\star$ and $\PP^\circ$ are established together with a closed form expression for the Radon-Nikodym derivative. It turns out that  
\[ \frac{\dd \PP^{\star,t}}{\dd \PP^{\circ,t}}(X^\circ)  =
 \frac{\tilde p(0,u)}{p(0,u)}\frac{p(t, X^\circ_{t}; T, v)}{\tilde p(t, X^\circ_{t}; T, v)} \: \exp\left( \int_0^t G(s,X^\circ_s) \dd s\right), \]
where $G$ is given by 
\begin{align}\label{eq:G} G(s,x) &= (b(s,x) - \tilde b(s,x))^\T \tilde r(s,x)\nonumber \\ & \qquad -  \frac12 \trace\left(\left[a(s, x) - \tilde a(s, x)\right] \left[\tilde H(s,x)-\tilde{r}(s,x)\tilde{r}(s,x)^\T\right]\right)
\end{align}
(Cf.\ proposition 1 in \cite{Schauer}). 
Upon taking the expectation and the limit $t\uparrow T$ it is proved in \cite{Schauer} that 
\begin{equation}\label{likeli}
\Phi(X^\circ) = \frac{\tilde{p}(0, u)}{p(0, u)}
 \exp\left( \int_0^T G(s,X^\circ_s) \dd s\right).
\end{equation}
This time the term $p(t, X^\circ_t; T,v)$ is essentially cancelled by $\tilde p(t, X^\circ_t; T,v)$ and henceforth disappears in the limit. From the expression of $\Phi$ we deduce that guided proposals are feasible. 

The class of linear processes,
\begin{align}\label{linsdehomog}
		\dd  \tilde X_t = \tilde B(t) \tilde X_t\dd t  + \tilde\beta(t)  \dd t +  \tilde \sigma(t)  \dd  W_t,
\end{align}
 is a flexible class  with known transition densities and its induced guided proposals  satisfy the conditions for absolute continuity derived in \cite{Schauer} under weak conditions on $\tilde{B}$, $\tilde\beta$ and $\tilde{\sigma}$. Proposal processes $X^\circ$ derived by choosing a linear process as in \eqref{linsdehomog} will be referred to as {\it linear guided proposals}. One key requirement for absolute continuity of $X^\star$ and $X^\circ$ is that $\tilde\sigma$ is such that $\tilde{a}(T)=(\tilde\sigma \tilde\sigma')(T)= a(T,v)$. 
A particularly simple type of guiding proposals is obtained upon choosing $\dd\tilde{X}_t = \tilde\beta(t) \dd t+ \si(T,v) \dd W_t$. For this particular choice
\begin{equation}\label{eq:simple-linear} b^\circ(t,x)= b(t,x)  + \frac{a(t, x) a(T,v)^{-1}}{T-t}\left(v-x- \int_t^T \tilde\beta(s) \dd s\right). \end{equation}
 Depending on the precise form of $b$ and $\si$ it can nevertheless be  advantageous to use guided proposals induced for non-zero $\tilde{B}$. In section \ref{sec:choiceguided} we discuss several strategies for choosing the process $\tilde{X}$.  
\begin{rem}
For  guided proposals,  it is easily seen that the process  $Z$ appearing in proposition \ref{prop:existence-g} satisfies \eqref{eq:def-mu} with 
$\mu(t,x)=\si^\T (t,x) \left(r(t,x) - \tilde{r}(t,x)\right)$.
\end{rem}

\begin{rem}
In case $b$ and $\si$ are of the forms $b(s,x)=\beta(s) + B(s) x$ and $\si(s,x)=\si(s)$, then we can trivially take $\tilde{b}=b$ and $\tilde\si=\si$. By equation \eqref{eq:G} it follows that in this case $\Psi\equiv 1$. This implies that $A_2$ in algorithm \ref{alg1} does not depend on $\{Z_i,\, i=1,\ldots, n\}$ and simulating diffusion bridges is unnecessary. That is, step 2 of algorithm \ref{alg1} can be omitted.  
\end{rem}

%--------------------------------------
\subsection{Drift-independent guided proposals}\label{sec:adjustedguided-proposals}

The proposals with $\la=1$ provided by \cite{DelyonHu} are a special case of guided proposals only in case $\si$  is constant. These are recovered upon choosing $\tilde{b}\equiv 0$ and $\tilde{\si}=\si$.  Proposals with $\la=0$ are a special case when both $b$ and $\si$ are constant and correspond to choosing $\tilde{b}=b$ and $\tilde{\si}=\si$. The latter type of proposals enjoys  quite some popularity in the literature, especially when discretised with the multiplicative  correction term added to the diffusion term introduced by \cite{DurhamGallant} (the resulting discrete time proposal is called the {\it modified diffusion bridge}, we get back to this in section \ref{sec:tc}). As such proposals are independent of the drift these can only work satisfactory if the drift in locally constant.

 In this article we do not aim to make a formal comparison of guided proposals and Delyon-Hu proposals.  Nevertheless, we wish to remark that for the latter class of proposals  both in case $\la=0$ and when $\la=1$ the resulting bridges  may not resemble true bridges. An illuminating example is given in the introductory section of \cite{Schauer} and we refer to that paper for further discussion on this rather subtle issue.  In case the reader is uncomfortable with the additional freedom for choosing the process $\tilde{X}$,  proposals similar (but not equal to) Delyon-Hu proposals can be obtained by taking $\dd \tilde{X}_t=\si(T,v) \dd W_t$, where $\si(T,v)\si(T,v)^\T=a(T,v)$. In that case we get proposals with
\[  b^\circ(t,x)= b(t,x)+ a(t,x) a(T,v)^{-1} \frac{v-x}{T-t}. \]
Proposals that ignore the drift completely can be  defined by \[ b^\circ(t,x)= a(t,x) a(T,v)^{-1} \frac{v-x}{T-t}, \]
We call these {\it drift-independent guided proposals}. 
The acceptance probability for drift-independent proposals can easily be obtained from \eqref{likeli} and equals
\begin{align*} \Phi(X^\circ)&= \frac{\tilde{p}(0,u)}{p(0,u)} \exp\left( \int_0^T G(s,X^\circ_s) \dd s+ \int_0^T b(s, X^\circ_s)^\T a^{-1}(s,X^\circ_s) \dd X^\circ_s\right. \\ & \left.  - \frac12 \int_0^T b(s, X^\circ_s)^\T a^{-1}(s,X^\circ_s)\left[b(s, X^\circ_s)+2 a(s,X^\circ_s) a(T,v)^{-1} \frac{v-X^\circ_s}{T-s}\right] \dd s \right),
\end{align*}
where $G$ is computed with $\tilde{b}\equiv 0$ and $\tilde{a}=a(T,v)$.

%--------------------------------------
\subsection{Choice of guided proposals}\label{sec:choiceguided}

In this section we discuss the choice of guided proposals. We propose the following strategies:
\begin{enumerate}
\item {\it Linearisation of the drift}. In some examples there is a natural point at which to linearise,  as in example \ref{ex:arctan}. If this is not the case, one can use a (weighted) regression, as explained in example \ref{ex:chemreaction}.  

\item {\it Solving the dynamical system associated to the SDE}. Suppose $x(t)$ satisfies the deterministic differential equation
\begin{equation}\label{nonoise}
\dd x(t) = b(t, x(t))\dd t,\quad x(0) = x_0.
\end{equation}
Then 
\begin{equation}\label{spacetotime}
\dd \tilde X_t = b(x(t)) \dd t + \tilde \sigma \dd W_t.
\end{equation}
 is clearly of the form \eqref{linsdehomog} with $\tilde\beta(t)=b(x(t))$,  $\tilde{B}\equiv 0$ and  $\tilde\sigma = \sigma(T,v)$.  This approach is illustrated in example \ref{ex:lotka}. 

\item {\it Combined approach}. Approximate $b(t,X_t)$ with $b(t,x(t)) +V(t,x(t)) (X_t-x(t))$, where $V(t,y)$ is the matrix with elements $V(t,y)_{i,j}=\partial b_i(t,y)\,/\, \partial y_j$ for $y \in \RR^d$.  This gives linear guided proposals with 
\[ \tilde\beta(t)= b(t,x(t)) - V(t,x(t)) x(t) \quad \text{and} \quad \tilde B(t)=V(t,x(t)). \]
This is closely related to the linear noise approximation of the SDE for $X$ as used in \cite{Whitaker}. 

\item {\it Iterative linearisation procedures}.
A further technique using ideas by \cite{Whitaker} is obtained by setting $\tilde\beta(t)=b(t, \expec{ \bar X^\star(t)})$ 
where $\bar X^\star$ is a tractable diffusion bridge from $u$ to $v$ (derived for example from a preliminary linear approximation to $X$). 

We will always have $\tilde\beta(0) = b(0,u)$ and $\tilde\beta(T) = b(T,v)$. Linear interpolation gives
\begin{equation}\label{eq:propbeta-simple}
\tilde\beta(t) =  (1-t/T) b(0,u) +  (t/T) b(T,v).
\end{equation}

\end{enumerate}

\begin{ex}\label{ex:arctan}
Let $X$ be the  diffusion process described by the SDE
\begin{equation}\label{eq:sde-arctan}
\dd X_t = (\alpha \arctan(X_t) + \beta) \dd t + \sigma \dd W_t.
\end{equation}
If $\alpha< 0$, $\tfrac{\pi}{2}\alpha < \beta< -\tfrac\pi2 \alpha $ this process is mean reverting to $\tan(-\beta/\alpha)$. For $x\approx \tan(-\beta/\alpha)$
\[b(x) \approx \alpha \cos^2(-\beta/\alpha) (x -\tan(-\beta/\alpha)).\]
So it makes sense to take linear proposals with 
\[\tilde{B}=\alpha \cos^2(-\beta/\alpha),\quad \tilde{\beta}= \tfrac12 \alpha \sin(2 \beta/\alpha)\quad \text{and} \quad \tilde\sigma=\sigma.\]
\end{ex}

%--

\begin{ex}\label{ex:chemreaction}
Here we consider a simple example in which the dynamics of a chemical reaction network are approximated by a system of stochastic differential equations.   Suppose we have four reactions among chemicals $A$, $B$ and $C$:
\begin{align*} 
 \emptyset & \stackrel{\th_1}{\rightarrow} A  \qquad \qquad
 A  \stackrel{\th_2}{\rightarrow} B \\
 A + B  & \stackrel{\th_3}{\rightarrow} C \qquad \qquad 2 C   \stackrel{\th_4}{\rightarrow} \emptyset
\end{align*}
The amount of the chemicals $A$, $B$, $C$ at time $t$ can be modelled as a pure jump Markov process which can subsequently be approximated by the  diffusion process $X_t \in \RR^3$ which solves the \emph{Chemical Langevin Equation} (\cite{Fuchs}, chapter 4)
\begin{equation}\label{prokarsde}	 \dd X_t = S h_\th(X_t) \dd t + S \operatorname{diag}  (\sqrt{h_\th(X_t)})  \dd W_t 
\end{equation}
driven by a $\RR^4$-valued Brownian motion. Here
\[ S = \begin{bmatrix}  1 & -1 & -1 & 0 \\
0 & 1 & - 1&0 \\
0 & 0 & 1 & -2\\
\end{bmatrix} \]
 is the stoichiometry matrix of the system describing the  chemical reactions
and $h_\th(x) =\th \circ h(x)$
is a function describing the hazard for a particular reaction to happen.
Here $\circ$ denotes the Hadamard (or entrywise) product of two vectors and
\[\th = [\th_1, \th_2, \th_3, \th_4]^\T \qquad h(x) =[1,  x_1,  x_1 x_2,  x_3(x_3-1)/2]^\T. \]

We choose $\tilde B$ and $\tilde \beta$ to depend on $\th$ (but not on time) so that $\tilde{B} x + \tilde\beta$  approximates $b_\th(x)$. While it is possible to take different approximations specifically tailored for each bridge segment, it is computationally advantageous to work with a global approximation to $b_\th$ (as we need to evaluate $\tilde p$ in the expression for $A_2$, see also the discussion in section \ref{sec:impl}). To this end, we replace $h$ by a linear approximation $\tilde{h}$ which allows for obtaining $\tilde{B}_\th$ and $\tilde\beta_\th$ from the equation
\[ \tilde{B}_\th x +\tilde\beta_\th = S (\th \circ \tilde{h}(x)). \]
As the first two components of  $h(x)$ are linear, we take $\tilde{h}_1(x)=h_1(x)$ and $\tilde{h}_2(x)=h_2(x)$. We approximate $h_3$ by $\tilde{h}_3(x) = c_3+ u_{3,1} x_1 + u_{3,2} x_2$. Values for $c_3, u_{3,1}$ and $u_{3,2}$ are obtained from a weighted linear regression of $x_1 x_2$ on $x_1$ and $x_2$, with weights proportional to $x_1 x_2$. Similarly, we take $\tilde{h}_4(x)= c_4+ u_{4,3} x_3$.  Values for $c_4$ and $u_{4,3}$ are obtained from a weighted linear regression of $\frac12 x_3(x_3-1)$ on $x_3$, with weights proportional to $x_3(x_3-1)$.
We take a {\it weighted} regression in this way because for a good proposal the error matters more if the corresponding dispersion component is small. 
For $\tilde\sigma$ we choose $\tilde\sigma=S \operatorname{diag}  \left(\sqrt{\tilde{h}_\th(x_i)}\right)$ on the segment between times $t_{i-1}$ and $t_i$. 

Note that this approach for constructing $\tilde{B}$ and $\tilde\beta$ can be applied generally to stochastic differential equations arising from chemical reaction networks.

\end{ex}

\begin{ex}\label{ex:lotka}
The Lotka-Volterra model with multiplicative noise (cf. \cite{KhasminskiiKlebaner}) is given by the Stratonovich stochastic differential equation 
\begin{align}\label{lotka}
\begin{split}
\dd X_t &= \left(\theta X_t - X_tY_t\right)\dd t + \sigma X_t\circ \dd W^{(1)}_t,\qquad X_0 = x_0\\
\dd Y_t &= \left(-\theta Y_t + X_tY_t\right)\dd t + \sigma Y_t \circ \dd W^{(2)}_t,\qquad Y_0 = y_0.
\end{split}
\end{align}
By  It\=o's formula, $(\xi_t, \eta_t) = (\log X_t, \log Y_t)$ satisfies
\begin{align*}
\dd \xi_t &= \left(\theta  - \e^{\textstyle  \eta_t}\right)\dd t + \sigma\dd W^{(1)}_t, \qquad \xi_0 =\log x_0\\
\dd \eta_t &= \left(-\theta  + \e^{\textstyle  \xi_t}\right)\dd t + \sigma\dd W^{(2)}_t\qquad \eta_0=\log y_0.
\end{align*}
Proposals for a bridge  that hits $(\xi_T, \eta_T)=(\log X_T, \log Y_T)$ at time $T$  can be derived from the deterministic dynamical system associated with \eqref{lotka}.  
The deterministic system $(x, y)$ has trajectories $xy \e^{-\frac{1}{\theta} (x + y)} = K$ with $K$ depending on $x_0, y_0$. The trajectory can be parametrised by 
\[
x(z) = \frac{z}{2} \pm \sqrt{z^2 - 4 K \e^{z / \theta}}, \quad  y(z) = z - x(z),
\]
where time is implicit and can be recovered from $z$ by the equation $\theta \sqrt{z^2  - 4K e^{z/\theta}}\dd t = \pm \dd z$ (Cf.\  \cite{SteinerGander}). We obtain guided proposals for $(\xi^\circ_t, \eta^\circ_t)$ by taking $\tilde{B}\equiv 0$ and $\tilde\beta(t)= (\th-x(t), -\th + y(t))^\T$. These proposals can subsequently be transformed to proposals for $(X^\circ_t, Y^\circ_t)$.
%In figure \ref{fig:lotka} we show that the proposals visually resemble true bridges reasonably well. 
%\begin{figure}
%\begin{center}
%\includegraphics[width=0.4\textwidth]{prop}
%\includegraphics[width=0.4\textwidth]{true}
%\end{center}
%\caption{Sample trajectories in the $(x,y)$-plane from the proposal $(X^\circ, Y^\circ)$ (left) and the true bridge (right) from $(x_0, y_0) = (0.3, 0.3)$ to  $(x_T, y_T) = (2.14, 1.75)$ at $T = 3.65$, $\theta=1$ and $\si= 0.15$. The  red dotted curve shows the periodic trajectory of the deterministic system starting at $(x_0, y_0)$.}
%\label{fig:lotka} 
%\end{figure}

\end{ex}

%%%%%%%%%%%%

\section{Numerical discretisation of guided proposals}\label{sec:tc}

Simulation of $X^\circ$ and numerical evaluation of $ \Psi(X^\circ) := \exp\left( \int_0^T G(s,X^\circ_s) \dd s\right)$  is numerically cumbersome since the drift of $X^\circ$ and the integrand  $G$ explode for $s$ near the endpoint $T$. 
\begin{ex}\label{ex:exploding}
Suppose $\si$ is constant and we take  $\tilde{X}= \si \dd W_t$. Then we have $\tilde{r}(s,x)= \tilde{a}^{-1} (v-x)/(T-s)$, where $\tilde{a}=\sigma \sigma^\T$. Hence the drift of the SDE for $X^\circ$ explodes when $s\uparrow T$.  Furthermore, 
\[	\log \Psi(X^\circ)=\int_0^T  b(s,X^\circ_s)^\T \tilde r(s,X^\circ_s) \dd s = \int_0^T  b(s,X^\circ_s)^\T  \tilde{a}^{-1} \frac{v-X^\circ_s}{T-s} \dd s, \]
which shows the integrand explodes as well. 
\end{ex}
In this section we explain how these numerical problems can be dealt with using a time change and scaling of the proposal process. The purpose is not solely obtaining a more accurate discretisation scheme for the SDE, but above all accurate  evaluation of the integral appearing in $\Psi(X^\circ)$. 

For the particular example just given \cite{Clark}   proposed to perform a time change and scaling of the proposal process to remove the singularities.  Define $\tau^C:[0,\infty) \to [0,T)$ by  $\tau^C(s) = T(1-\e^{-s})$ and  $U^C_s = \e^{s/2}(v -X^\circ_{\tau^C(s)})$. Then $U^C$ satisfies the stochastic differential equation
\begin{align*}
\dd U^C_s =  \;&   - T\e^{-s/2}b(T(1-\e^{-s}), v -\e^{-s/2} U^C_s  ) \dd s -\frac12  U^C_s \dd s -\sqrt{T} \sigma \dd W_s,  
\end{align*}
which behaves like a zero-mean  mean-reverting Ornstein-Uhlenbeck process as $s\to \infty$. Furthermore, 
\[  \log \Psi(X^\circ)= \int_0^\infty   e^{-s/2}b(\tau^C(s),v-e^{-s/2} U^C_s)^\T T \tilde{a}^{-1}   U_s^C  \dd s  
\]
(note that there are some minor typographical errors in \cite{Clark}). Clearly, if $b$ is bounded,  this removes the singularity near $T$, but at the cost of having to deal  with an infinite integration interval. For this reason, we propose a different time-change and scaling.

The time change and scaling due to \cite{Clark} is a special case obtained from considering the process $U_s=m(s)\left(v(\tau(s))-X^\circ_{\tau(s)}\right)$, where $s\mapsto \tau(s)$ is nondecreasing. The choice by \cite{Clark} corresponds to  $\tau(s)=T(1-e^{-s})$ and $m(s)=e^{s/2}$. In the following we denote the time derivatives of $m$ and $\tau$ by $\dot m$ and $\dot \tau$ respectively. 
The time changed process $U=(U_s,\, s\in [0,T))$ satisfies the stochastic differential equation
\begin{multline*}
\dd U_s =  \left(\frac{\dot{m}(s)}{m(s)} U_s  -m(s) \dot\tau(s) b^\circ(\tau(s),v-U_s/m(s))\right) \dd s \\   -m(s) \sqrt{\dot \tau(s)} \si(\tau(s), v-U_s/m(s)) \dd W_s
\end{multline*}

Using the setting of example \ref{ex:exploding}, we motivate another choice of $\tau$ and $m$ for improving numerical accuracy. For the example,  the drift of $U$ is given by 
\[ -m(s)\dot\tau(s) b(\tau(s), v-U_s/m(s)) + \left(\frac{\dot m(s)}{m(s)} - \frac{\dot \tau(s)}{T-\tau(s)} \right) U_s \]
and  $\log \Psi(X^\circ)$ can be expressed in terms of $U$ as follows
\begin{equation}\label{eq:integextc}  \int_{\tau(0)}^{\tau(T)} b(\tau(s),v-U_s/m(s)) \tilde{a}^{-1} \frac{U_s}{m(s)} \frac{\dot \tau(s)}{T-\tau(s)} \dd s . \end{equation}
 As shown in \cite{Schauer}, up to a logarithmic term, $v-X^\circ_s \sim \sqrt{T-s}$ for $s$ close to $T$. Therefore,
$U_s \sim m(s) \sqrt{T-\tau(s)}$ which implies that the possibly exploding part of the integral in \eqref{eq:integextc} satisfies
\[    \frac{U_s}{m(s)} \frac{\dot \tau(s)}{T-\tau(s)} \sim \frac{\dot\tau(s)}{\sqrt{T-\tau(s)}}. \]
To make this constant, we take $\tau(s)=s(2-s/T)$. 
Furthermore, we choose $m(s)=1/(T-s)$ (see section \ref{sec:motivation-scaling} for a justification). 
With these choice of $\tau$ and $m$, $U$ satisfies the SDE
\[ \dd U_s =-\frac2{T} b(\tau(s),v-(T-s) U_s) \dd s - \frac1{T-s} U_s\dd s - \sqrt{\frac{2}{T}} \frac1{\sqrt{T-s}} \si \dd W_s,\quad U_0=\frac{v-u}{T}. \]
Compared to the original SDE for $X^\circ$, we see that an additional exploding factor appears in the diffusion coefficient. At first sight, this may seem like we have worsened the numerical problems. Note however that the integral we wish to evaluate ($\log \Psi(X^\circ)$) behaves much better now. For $s\approx T$, the process $U$ behaves like a mean-zero stationary Ornstein-Uhlenbeck process, with balanced increased mean-reversion and diffusivity. The process $U^C$ proposed by \cite{Clark} behaves like an Ornstein-Uhlenbeck process for large times as well, and we see that with our choice of $\tau$ we speed up time to run through this process much faster, preventing us from evaluating an integral over an unbounded integration region. 

%%%%
\subsection{Time changing and scaling of linear guided proposals}

Based on the motivational derivations of the preceding section, we define a convenient time change and scaling in this section.  
To do this, we need a few more results from \cite{Schauer}. If $\tilde{X}$ is a linear process (satisfying equation \eqref{linsdehomog}), then
\begin{equation}\label{eq:rel-r-H}
\tilde r(s,x) = \tilde H(s) (v(s)-x),
\end{equation}
where
\begin{equation}\label{def_v}v(s) =  F(s,T) v - \int_s^T F(s,z) \tilde \beta(z) \dd z
\end{equation}
($\tilde r$ and $\tilde H$ are defined in equation \eqref{eq:notation}). 
Here  $F(t,s)=F(t)F(s)^{-1}$ with $F(t)$   the fundamental $d\times d$ matrix that satisfies
\begin{equation}\label{def:F} F(t) = \I + \int_0^t \tilde B(z) F(z) \dd z. \end{equation}
Define the process $U$ by 
\begin{equation}\label{U}
	U_s :=\frac{v(\tau(s))-X^\circ_{\tau(s)}}{T-s}.
\end{equation}
This implies
\begin{equation}\label{eq:Xcirc-back}
 X^\circ_{\tau(s)} = 	 v(\tau(s))- (T-s) U_s=: \Gamma(s, U_s).
\end{equation}

\begin{lemma}\label{lem:tc}
The time changed process $U=(U_s,\, s\in [0,T))$ satisfies the stochastic differential equation
\begin{equation}\label{Usde}
\begin{aligned}
\dd U_s = &  \frac2T\dot v(\tau(s))\dd s  - \frac{2}{T} b(\tau(s), \Gamma(s, U_s)) \dd s  \\&+ \frac{1}{T-s}\Big(\I - 2a(\tau(s), \Gamma(s, U_s))J(s)\Big)  U_s  \dd s \\
& - \sqrt{\frac{2}{T}} \frac1{\sqrt{T - s}} \sigma(\tau(s),\Gamma(s, U_s) ) \dd W_s, \qquad U_0=\frac{v-u}{T}
\end{aligned}
\end{equation}
where $W$ is a Brownian motion and $J$ defined by
\begin{equation}\label{eq:defJ}
J(s) = \tilde H(\tau(s))(T-\tau(s))
\end{equation}
satisfies  $ \lim_{s \uparrow T} \tilde a(s) J(s) = \I$. 
Moreover,
 \begin{align*}\int_0^T G(s,X^\circ_s) \dd s &= 2\int_{ 0}^{ T}  (b - \tilde b)^\T (\tau(s),\Gamma(s, U_s))  J(s) U_s  \dd s \\ - &
  \int_{ 0}^{T} \trace\left[\frac{( a - \tilde a)(\tau(s),\Gamma(s, U_s))}{T-s} J(s)\left(\I-T\, U_s U_s^\T J(s)\right)  \right]\dd s.
\end{align*}
\end{lemma} 

If we simulate $U$ on an equidistant grid we can recover $X^\circ$ on a non-equidistant grid from equation \eqref{eq:Xcirc-back}. This implies $X^\circ$ is evaluated on an increasingly finer grid as $s$ increases to $T$. 
In our implementation, all computations are done in time-changed/scaled domain, and the mapping $g$ is in fact defined  by setting $U = g(\th, Z^\circ)$, where $Z^\circ$ is the driving Brownian Motion for $U$.

\subsection{Numerical illustrations}

In this section we present results of simulations to assess the decrease in discretisation error using the proposed time change and scaling. As a comparison, we also consider various alternative discretisation schemes. In all cases, we use the equidistant grid by imputing $m-1$ points on $[0,T]$ for discretisation. Define $h=T/m$ and set  $t_j=jh$, $j=0,\ldots, m$. The alternatives we consider are:
\begin{enumerate}
\item Euler discretisation of the SDE for $X^\circ$.
\item The Modified Diffusion Bridge (MDB) discretisation  introduced in \cite{DurhamGallant}. This discretisation is obtained by  applying Euler discretisation to the SDE for $X^\circ$ and adding a correction term to the diffusion coefficient. This gives the scheme $\{\vt{X}^\circ_{t_j}\}$ where 
\begin{equation}\label{eq:mdb}
 \vt{X}^\circ_{t_{j+1}} = \vt{X}^\circ_{t_j} + b^\circ(t_j, \vt{X}^\circ_{t_j}) (t_{j+1}-t_j) + \si(t_j, \vt{X}^\circ_{t_j})  \sqrt{\frac{T-t_{j+1}}{T-t_j}} (W_{t_{j+1}}-W_{t_j}) 
\end{equation}
\item Euler discretisation of the SDE for the time-changed process using $\tau$, but {\it without} the scaling. This means that we apply Euler discretisation to the SDE 
\[ \dd V_s = b^\circ(\tau(s),V_s) \dot\tau(s) \dd s + \sqrt{\dot\tau(s)} \si(\tau(s),V_s) \dd W_s,\qquad V_0=u \]
where $V_s=X^\circ_{\tau(s)}$.
\end{enumerate}
The first two of these schemes have gained quite some popularity in the literature, the third one is included to assess the effect of including a scaling. 

Within the simulation study, we considered $b\equiv 0$  and  $b(x)=-\arctan(x)$. 
We considered two types of guided proposals $X^\circ$:
\begin{enumerate}
\item proposals generated by choosing $\dd \tilde{X}_t = \si(T,v) \dd W_t$ which gives pulling term \[  \tilde{r}(t,x)=\si(T,v)^{-2} (v-x)/(T-t) \tag{\text{BM-pull}} \]
\item proposals generated by choosing $\dd \tilde{X}_t =-\beta \tilde{X}_t \dd t + \si(T,v) \dd W_t$, which gives pulling term
\[	\tilde{r}(t,x)= \frac{2 \beta}{\si^2(T,v)} e^{-\beta(T-t)}\frac{v-xe^{-\beta(T-t)}}{1-e^{-2\beta(T-t)}}.\tag{\text{OU-pull}}\] 
\end{enumerate}
In the simulation study we are  interested in accurate discretisation of the likelihood given in equation \eqref{likeli} which appears in the acceptance probabilities of the algorithms of Section \ref{sec:algo} (whether the considered pulling terms are good choices is of minor importance for that purpose). More precisely, we evaluate the discretisation of the path-integrals 
 $I(X^\circ)= \int_0^T G(s,X^\circ_s) \dd s$ in case of discretisation of the SDE for $X^\circ$, $I(V)=\int_0^T G(\tau(s), V_s) \dot\tau(s) \dd s$ in case of discretisation of the SDE for $V$ and $I(U)=\int_0^T G(\tau(s), v(\tau(s)-(T-s) U_s) \dot\tau(s) \dd s$ in case of discretisation of the SDE for $U$. 
At the finest discretisation level, we divide $[0,T]$ into  $2^L$ intervals of equal length. 
If $h=T/2^L$, then $t_j=jh$, $j=0,\ldots, 2^L$. We start by simulating on the finest grid. Next we redo the simulation on the grid of length $2^{L-1}$ {\it using the same Wiener process increments}. This can be continued iteratively (until there are only 2 intervals of equal length). 
The simulation study was run as follows:
\begin{compactenum}
\item Generate the sequence $\{t_j\}$ with $h=T/2^L$.
\item Generate Wiener increments on the generated grid.
\item Simulate a realisation $X^{\circ,L}$  of the diffusion bridge with the generated Wiener increments. Compute and store  $I(X^{\circ,L})$. 
\item for k=L-1 downto 2
\begin{compactitem}
\item Coarsen the grid by removing the 2nd, 4th, 6th, etc point from the grid and aggregate the Wiener-increments. Simulate a realisation $X^{\circ,k}$ of the diffusion bridge with these Wiener increments.
\item Compute the error $e_k=I(X^{\circ,k})-I(X^{\circ,L})$. 
\end{compactitem}
\item Repeat steps 1 up till 4 $B$ times and compute  the Root Mean Squared Error of all errors using an equal number of grid-points. 
\end{compactenum}
We chose $L$ sufficiently large such that the approximation for $I$ is virtually the same for all discretisation methods. 
As quadrature rule we used  the midpoint rule, where the integrand is evaluated at the left-point. 
\begin{figure}
\begin{center}
\includegraphics[scale=.9]{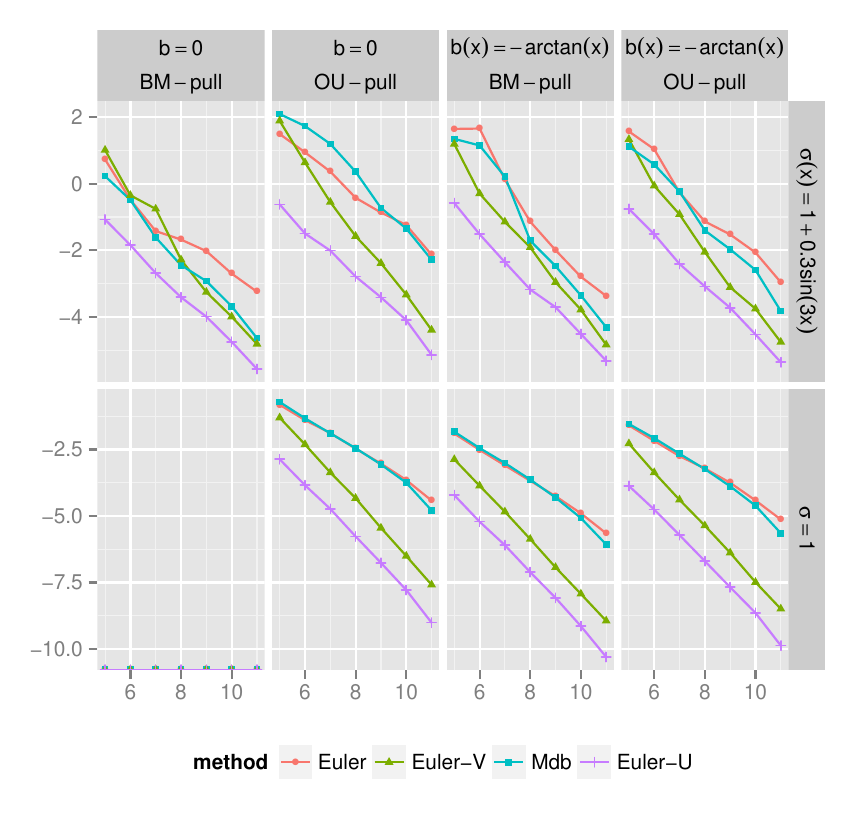}
\caption{$\log_2(RMSE)$ of $I(\cdot)$ versus $\log_2(\text{nr of segments})$. } %
\label{fig:discretisation1}
\end{center}
\end{figure}

In the simulations, we simulated bridges starting in $u=0$ at time $0$ and ending in $v=3$ at time $T=1$.
The results of the simulations are in figure \ref{fig:discretisation1} where we considered both $\si=1$ and $\si(x)=1+0.3\sin(3x)$. By definition, there is no error in the lower-left panel.  From the simulation results we see that for various combined choices of drift, diffusion coefficient and pulling term,  our approach performs best. We have run simulations with other values for $v$, $T$, $b$ and $\si$ leading qualitatively to the same conclusion.

 The beneficial effect of the time change and scaling is further illustrated in the examples of Sections \ref{sec:onedim-example} and  \ref{sec:fitz}.

\subsection{Order of convergence}

Ideally, one would derive a result on the order of convergence of each of the discussed discretisation methods for approximating $I$. We feel that this is outside the scope of this paper. As noted in \cite{PapaRobertsStramer} (page 676): ``Quantitative results on the relative efficiency of discretisation schemes are scarce in the literature.'' In case $b=0$ and $\si$ is constant (which is the simplest case to consider), \cite{PapaRobertsStramer} show that the strong order of convergence of the Euler scheme is $O(\sqrt{\delta})$ at $T-\delta$. This shows that the usual higher $O(\delta)$ strong order (which holds for diffusions with additive noise) is lost due to the exploding behaviour of the drift. 
Along similar lines as in \cite{PapaRobertsStramer}   one can prove that the strong order $\delta$ is maintained if the time-change and scaling is used. Admittedly, this is a rather weak result since {\it (i)} the case $b=0$ and $\si$ constant is very specific, {\it (ii)} the focus is on accurate evaluation of a path integral of the proposal bridge and not solely the process at specified points. The concept of strong order is not really needed here: we are interested in almost sure convergence of Euler approximation pathwise.   Under local Lipschitz conditions on the drift and diffusion coefficients, the pathwise convergence rate of the Euler scheme coincides up to an arbitrarily small $\eps>0$ with its strong convergence rate $1/2$ (Cf.\ \cite{Gyongy}). We expect the same pathwise convergence rate to hold for the integrals, when approximated using the proposed time-change and scaling.  In this sense, it is not unexpected that the lower panel in figure \ref{fig:discretisation1} shows lines with  slopes close to either $1/2$ (Euler, Mdb) or $1$ (Euler-V, Euler-U).

\subsection{Motivation for the scaling}\label{sec:motivation-scaling}
Consider the SDE for $\tilde{X}$ as defined in equation \eqref{linsdehomog}. The corresponding fundamental matrix is given in equation \eqref{def:F}. 
Define the process $\tilde{X}^\star$ as the process $\tilde{X}$, conditioned on $\tilde X_T = v$. Then $\tilde{X}^\star$ is a linear process itself with drift  $\tilde b^\star(t,x) = \tilde{B}(t) x + \tilde\beta(t) + \tilde{a}(t) \tilde{H}(t)  (v(t) -x)$ and diffusion coefficient $\tilde{\si}^\star(t)=\si(t)$.  
Denote the corresponding fundamental matrix  by $F^\star$. Hence $F^\star$ satisfies
\[
\frac{\dd }{\dd t} {F^\star}(t) =  \left( \tilde{B}(t)  - \tilde a(t) \tilde H(t)\right) F^\star(t),\qquad F^\star(0)=\I.
\]

\begin{thm}\label{thm:tcou}
Fix a nondecreasing differentiable mapping $\tau\,:\,[0,T] \to [0,\infty)$. 
If we define the scaling matrix $m$ by 
$ m(s) =  (T-s){F^\star}(\tau(s))^{-1}$. 
then the process $U$ defined by 
\[
U_s =   m(s) \left[v(\tau(s)) -  X^\circ_{\tau(s)} \right]
\]
(with $s\mapsto v(s)$ as defined in equation \eqref{def_v}) satisfies the SDE
\begin{multline*}
\dd U_s = \left( -\frac{U_s}{T-s} - m\dot\tau\left[b(\tau, \Gamma)-\tilde b(\tau, \Gamma) \right.\right.\\ \left.\left.   + ( a(\tau, \Gamma)-\tilde{a}(\tau)) \tilde H(\tau) m^{-1} U_s\right] \right)\dd s - m\sqrt{\dot\tau} \si(\tau, \Gamma)\dd W_s,
\end{multline*}
where   $\Gamma\equiv \Gamma(s,U_s)=v(\tau(s))-m(s)^{-1} U_s$. To lighten the notation we have written $\tau$, $\dot\tau$ and $m$ to denote $\tau(s)$, $\dot\tau(s)$ and $m(s)$ respectively.
\end{thm}
The proof is deferred to the appendix (section \ref{app:proofs}). 
\begin{cll}\label{cor:euler-good}
Let $\bar U_{t_i}$ denote the Euler approximation at time $t_i$ of $U$. 
If $a(t,x) \equiv a(T,v)=\tilde{a}$ and $b(t,x)=\tilde{b}(t,x) = \tilde{B}(t) x + \tilde\beta(t)$, then 
\[
\expec{U_{t_i} \mid U_{t_{i-1}} = u} = \expec{\bar U_{t_i} \mid \bar U_{t_{i-1}} = u}. 
\]
\end{cll}
\begin{proof}
In this case 
\[	\dd U_s = -\frac{U_s}{T-s}\dd s - m\sqrt{\dot\tau} \si(\tau, \Gamma)\dd W_s.
\]
Hence 
\[	\expec{U_{t_i} \mid U_{t_{i-1}} = u}=\frac{T-t_i}{T-t_{i-1}} u . \]
 It is easy to see that this coincides  with $\expec{\bar U_{t_i} \mid \bar U_{t_{i-1}} = u} $.
\end{proof} 

This shows that if we use linear guided proposals and use the scaling matrix $m$ defined in theorem \ref{thm:tcou}, then the Euler approximation of the process $U$ has the correct conditional expectation when $X$ itself is a linear process. Note that this is not necessarily the case without applying the scaling.

In case $\tilde\beta=0$, $\tilde{B}=0$ and $\tilde\si(t)=\tilde\si$, we have $F^\star(t)=\I/(T-t)$ and $m(s)=1/(T(T-s)) \I$. This means that we should have $m(s)=O(1/(T-s))$ for $s\approx T$.

%%%%%%%%%%
\section{Computational costs and implementation}\label{sec:impl}

In this section we discuss the computational cost of using guided proposals. For comparison, we add the computational cost of Delyon-Hu type proposals. Here we only consider the cost of imputation by diffusion bridges (including the computation of their acceptance probabilities). 
 Let
\begin{itemize}
\item $K$ denote the number of iterations of the data-augmentation algorithm;
\item $M$ denote the number of segments for imputations (so $M+1$ is the number of discrete-time observations);
\item $N$ denote the number of Euler-step applied to each segment.
\end{itemize}
The computational costs of simulating proposals are summarised in table \ref{table:comp}. We give some elucidation on this table. 
\begin{enumerate}

\item Applying guided proposals with $\tilde{B}\equiv 0$ gives minor additional computations compared to Delyon-Hu type proposals. One merely needs to compute $\int_t^T \beta_\th(s) \dd s$ on the whole augmented grid during all simulations. If $\beta$ does not depend on $\th$ this computation  needs to be carried out only once on the whole grid. 

\item If $\tilde{B}\not\equiv 0$, then simulation of  $U$ as defined in equation \eqref{Usde} requires evaluation of both $\dot{v}$ and $J$, where $v$ and $J$ are defined in equations \eqref{def_v} and \eqref{eq:defJ} respectively.  As $\dot{v}(s)=\tilde{B}(s) v(s) +\tilde\beta(s)$, evaluating $\dot{v}$ requires evaluation of $v$. This in turn requires evaluation of matrix exponentials. For evaluating $J$, we first compute  $\lambda$ as the solution to continuous time Lyapunov equation $B \lambda + \lambda  B^\T = -\tilde{a}$.  Using $\la$ we can evaluate $J$ using \[
J(s)  = z(s) \left(e^{- \tilde{B}_\th z(s)} \lambda e^{-\tilde{B}_\th^\T z(s)} -\lambda  \right)^{-1},
\]
where $z(s)=T-\tau(s)$. These functions  need to be computed  on the whole augmented  grid  in each iteration. In case $\tilde{B}$ does not depend on $\th$, both $J$ and $\dot{v}$ can be precomputed on a grid in advance to the MCMC-algorithm, preventing multiple expensive matrix exponential computations. 

\end{enumerate}

Besides simulation of the proposals, an acceptance probability needs to be computed. This requires evaluation of certain integrals of the proposal. A potential disadvantage of Delyon-Hu type proposals is that inverses appear. Moreover, stochastic integrals need to be approximated.

{
\newcommand{\smallcdot}{}
\begin{table}
\label{table:comp}
\centering \rara{1.3}
	\small\begin{tabular}{lccc}\toprule 
		&  $\int_{t}^{T} \beta_\th \dd s$ & $\exp\left(s B\right)$ & $B \lambda + \lambda  B^\T = \tilde{a}$ \\
		\cmidrule{1-4}
		Delyon-Hu  &  0 & 0 & 0\\
		$\tilde b(t,x) = \beta_\theta(t)$ & $N\smallcdot M\smallcdot K$ & 0 & 0\\
		$\tilde b(t,x) = B x + \beta_\theta(t)$  &$N\smallcdot M \smallcdot K $& $N$ & $M$\\ 
		$\tilde b(t,x) = B_\theta x + \beta_\theta(t)$ &$N\smallcdot M \smallcdot K $& $N\smallcdot K$ & $M  \smallcdot K$\\ 
	\end{tabular}
\caption{Overview of computational cost for simulating proposals.}
\end{table}	
}

%%%%%%%%%%
\section{Examples}\label{sec:examples}
The source code of the examples is available online.\footnote{See {\tt https://github.com/mschauer/BayesEstDiffusion.jl}.} It is  written in the programming language Julia (\cite{julia}).

\subsection{Example for one-dimensional diffusion}\label{sec:onedim-example}
In this section we discuss example \ref{ex:arctan}. The goal is twofold: {\it (i)} to show that the proposed algorithm does not deteriorate when  increasing the number of imputed points, {\it (ii)}: to show that the discretisation scheme of section \ref{sec:tc}  reduces discretisation error. 

We take the  diffusion process with dynamics of \eqref{eq:sde-arctan}.
Assume that we observe $X$  at times points $t=0, 0.3, 0.6\dots, T=30$ and wish to estimate $(\alpha, \beta, \sigma)$. As true values we took $\alpha = -2, \beta = 0$ and $\sigma = 0.75$. For generating the discrete time data  we simulated the process on $[0,T]$ at $400\,001$ equidistant time points using the Euler scheme and take a subsample. 

For  $\alpha$ and $\beta$ we chose apriori independently a $\scr N(0, \xi^2)$-distribution with variance $\xi^2 = 5$. For $\log \sigma$ we used an uninformative flat prior. We applied algorithm \ref{alg2} with $\rho=0$ in \eqref{eq:rw-innovation} with random walk proposals for $q(\sigma^\circ \mid \si)$ of the form
$\log\sigma^\circ :=  \log \sigma + u$ with $u \sim \scr U(-0.1, 0.1)$.

We initialised the sampler with $\alpha = -0.1$, $\beta = -0.1$ and $\sigma=2$ and varied the number of imputed points over  $m = 10, 100$ and $1000$. Acceptance rates for proposed bridges were in all cases between $94\%$ and $95\%$ and  for  $\si$ between $72\%$ and $73\%$.

\begin{figure}
\centering
\includegraphics[width=0.45\textwidth]{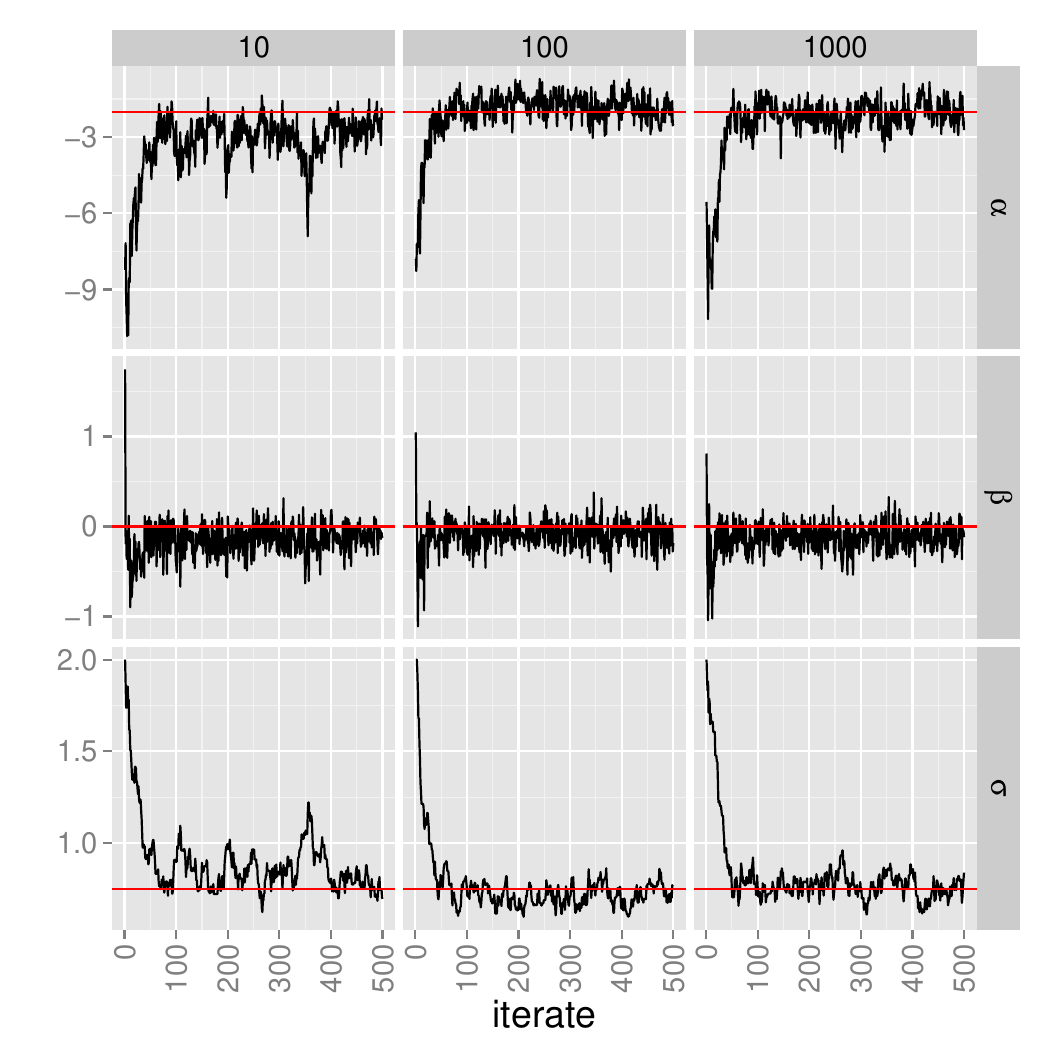}
\includegraphics[width=0.45\textwidth]{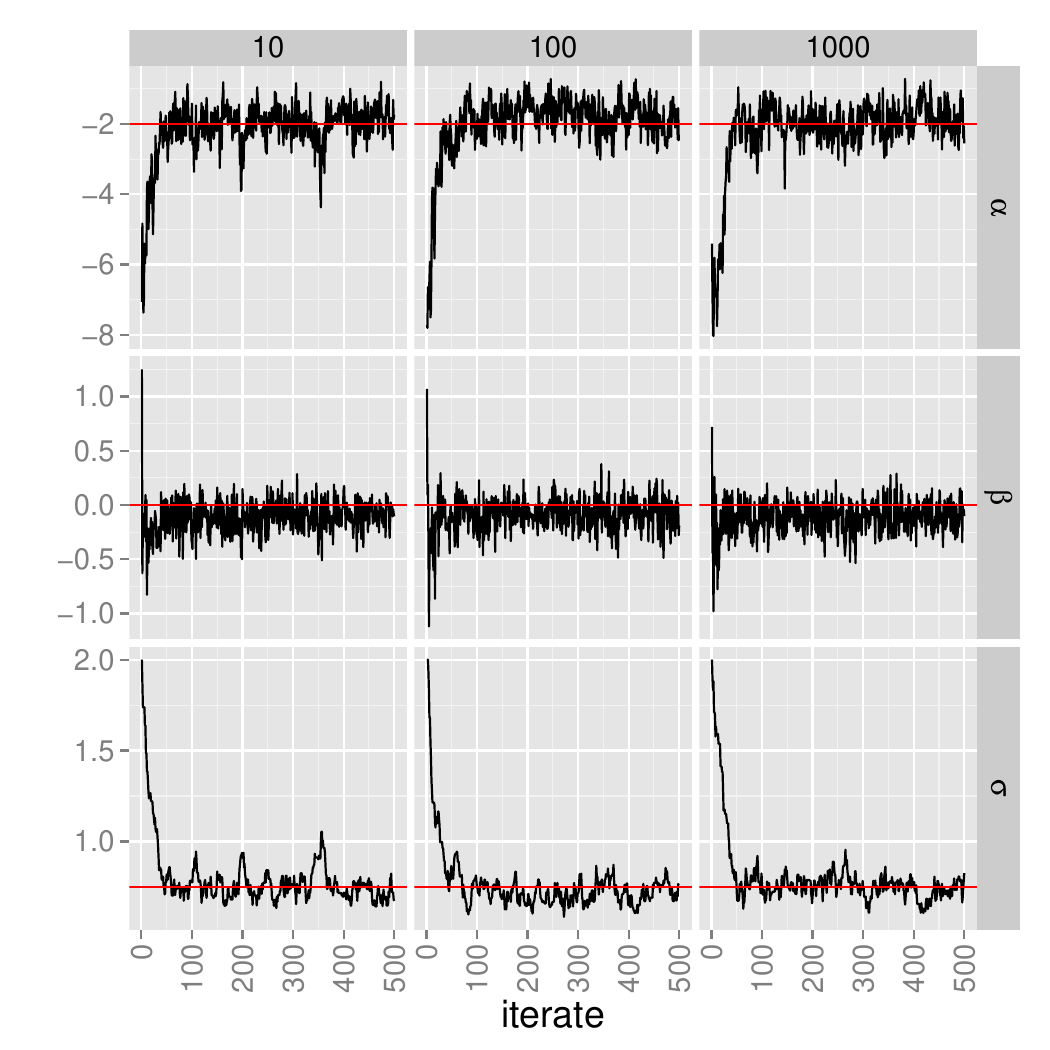}
\includegraphics[width=0.45\textwidth]{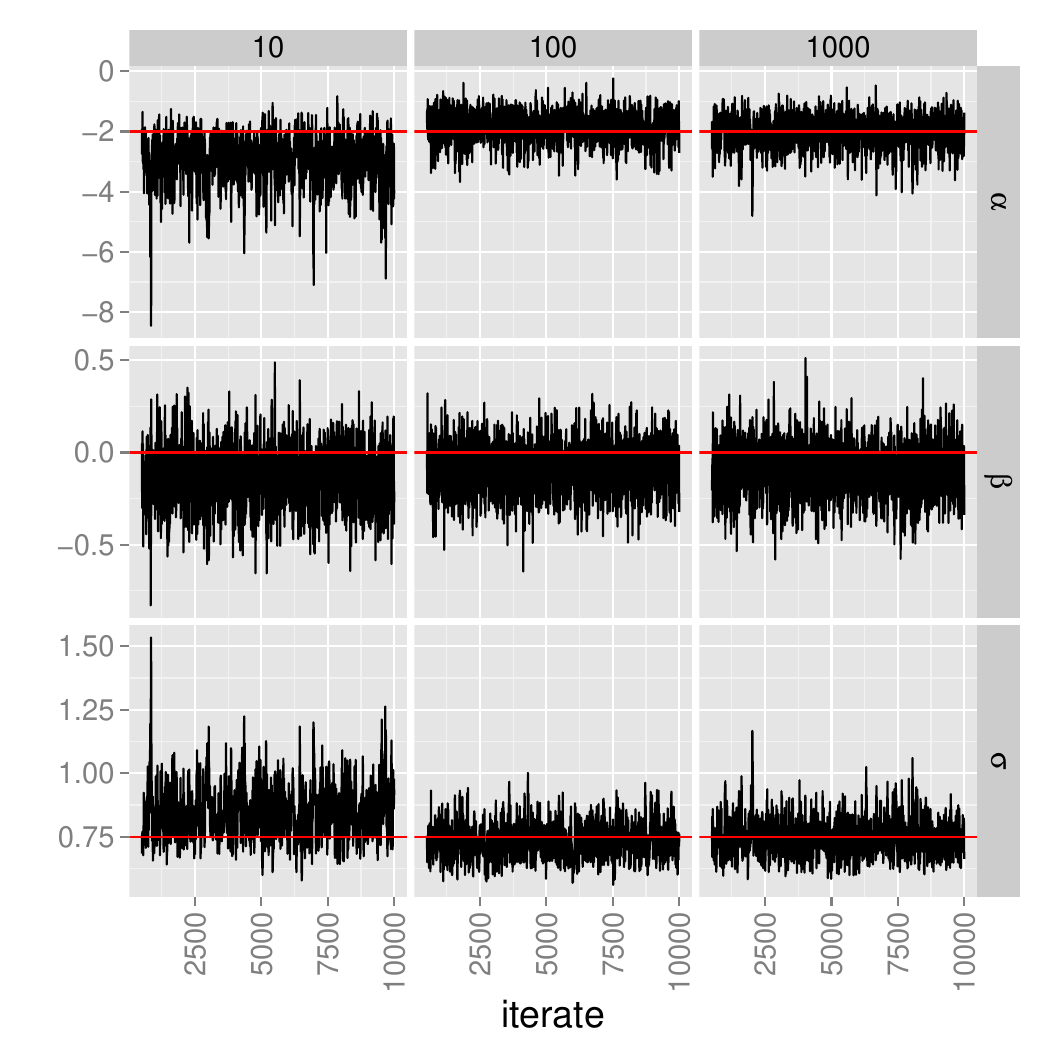}
\includegraphics[width=0.45\textwidth]{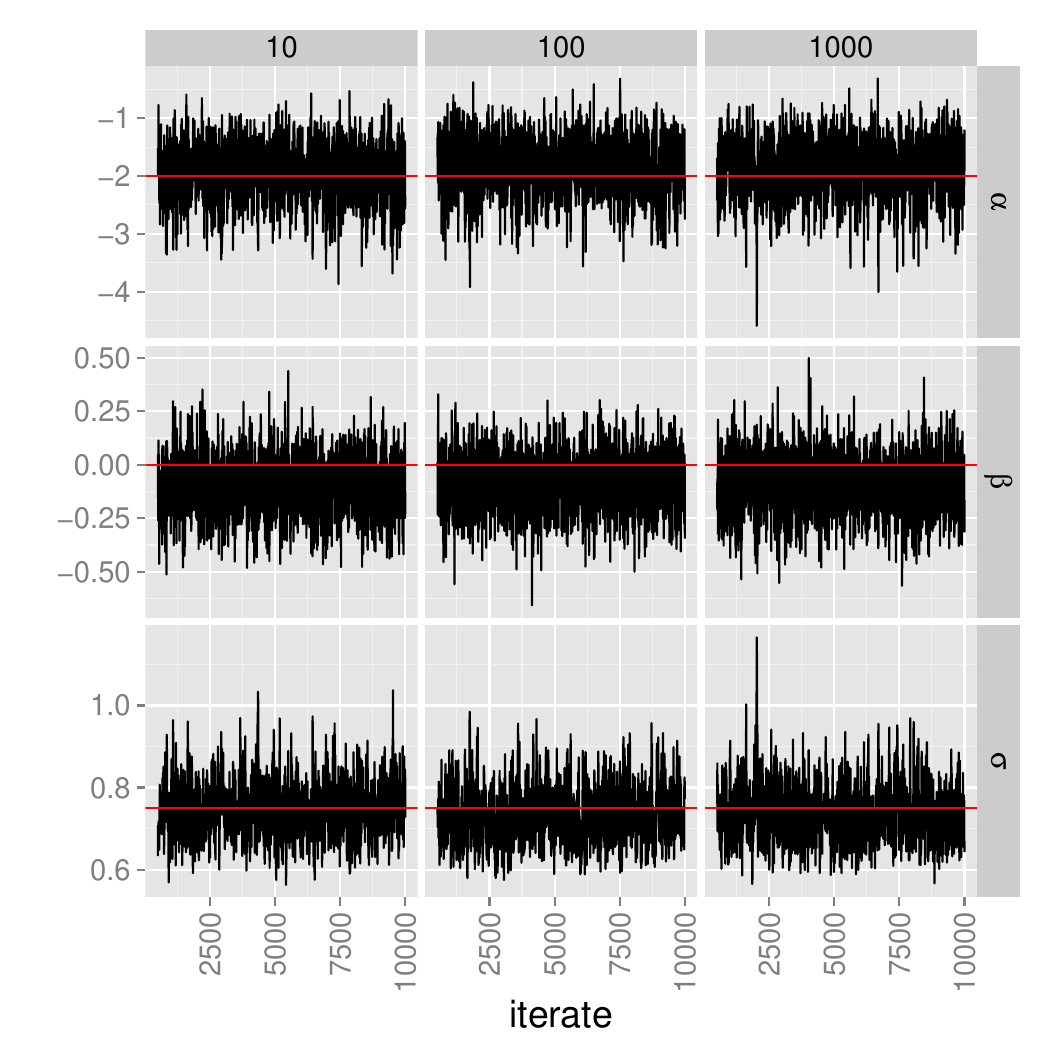}
\includegraphics[width=0.45\textwidth]{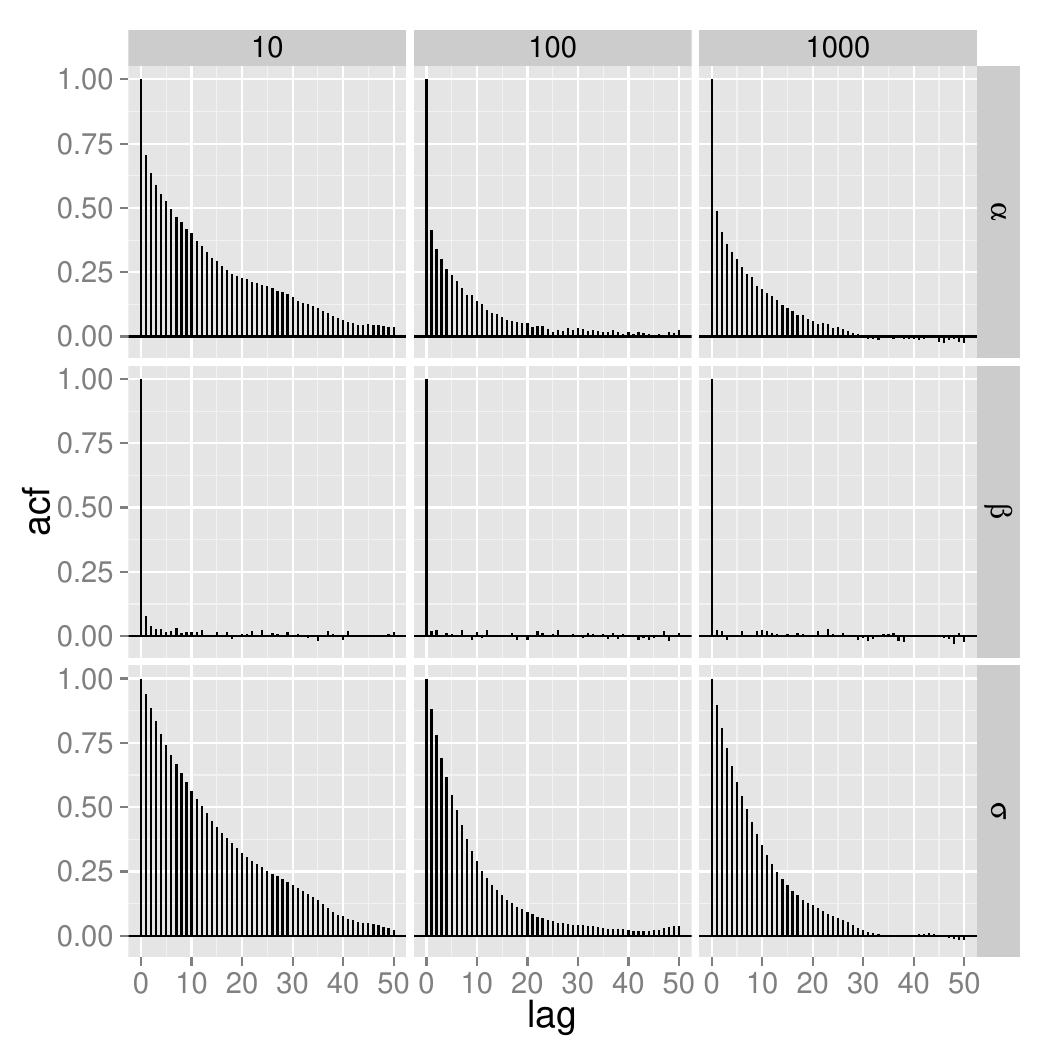}
\includegraphics[width=0.45\textwidth]{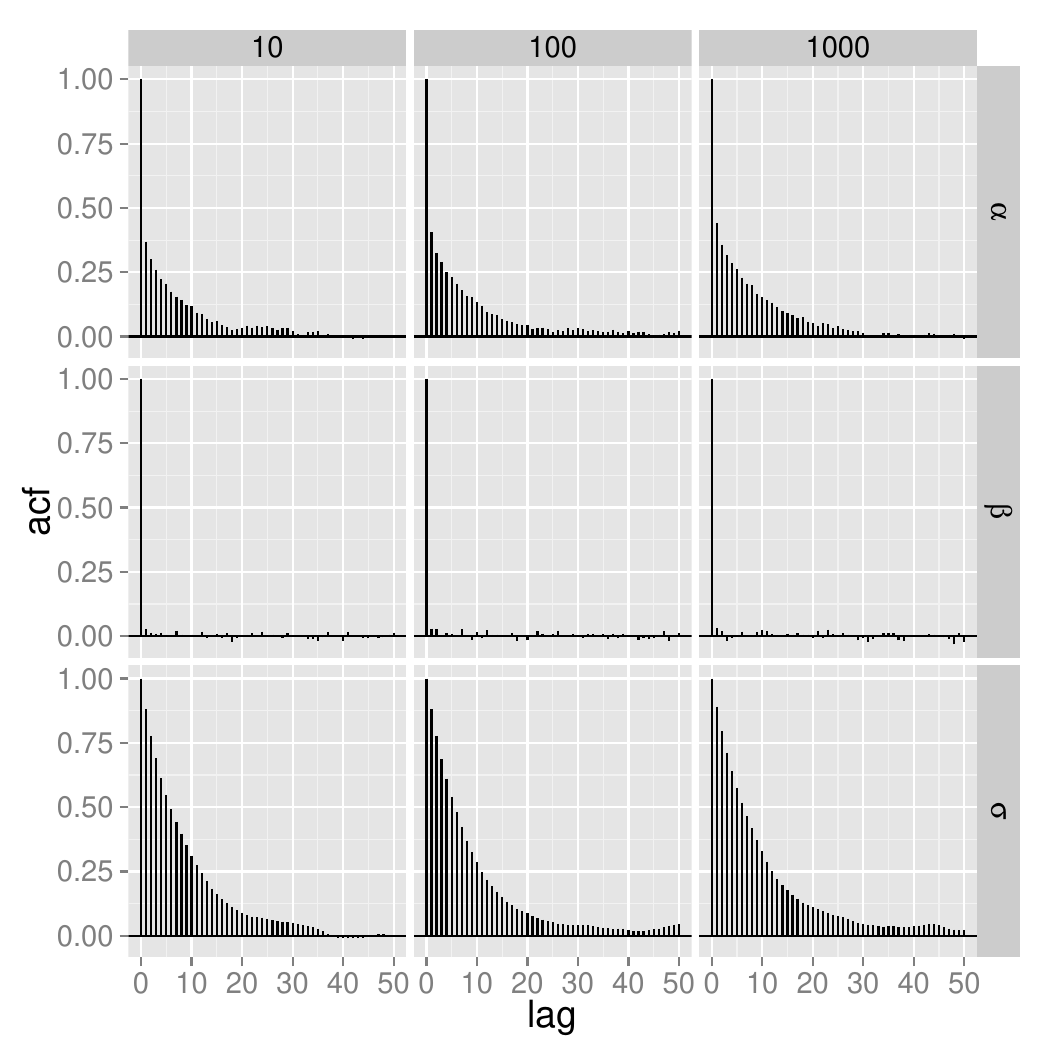}
\caption{Panels comparing different numbers of imputed points ($m = 10, 100, 1000$). Left: without time change. Right: with time change. Top: first 500 iterates. Middle: iterates 501-10.000. Bottom: ACF-plots based on iterates 501-10.000.}
\label{bigpanel1}
\end{figure}

\begin{figure}
\centering
\includegraphics[width=0.3\textwidth]{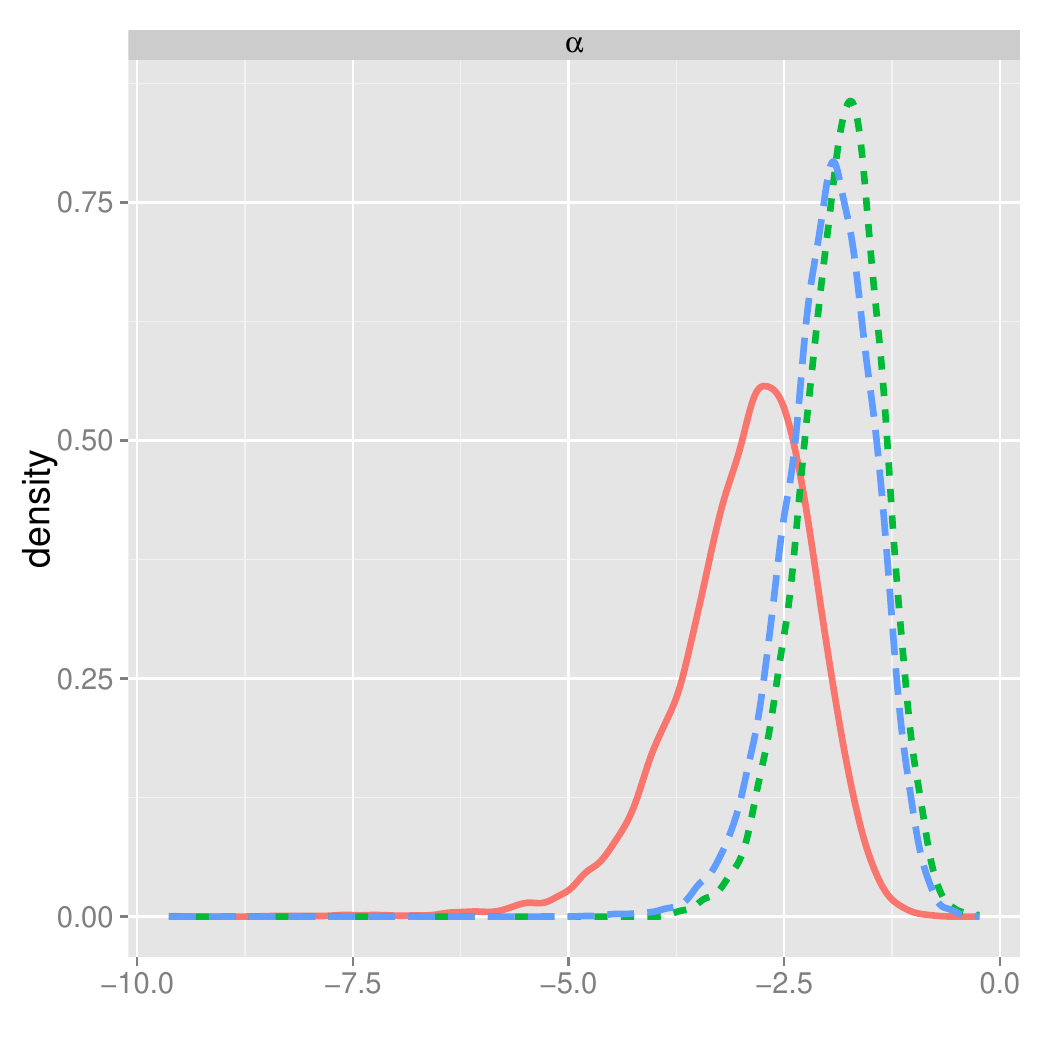}
\includegraphics[width=0.3\textwidth]{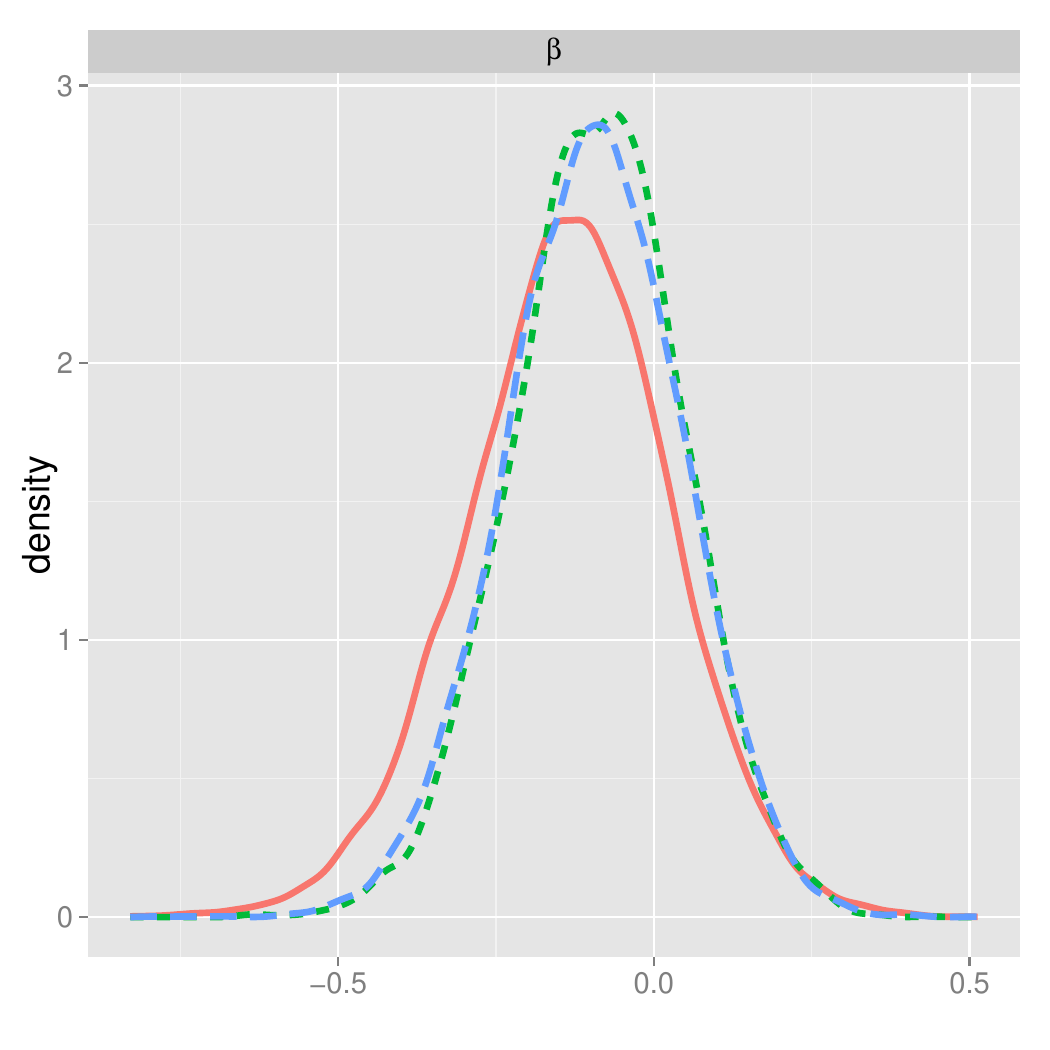}
\includegraphics[width=0.3\textwidth]{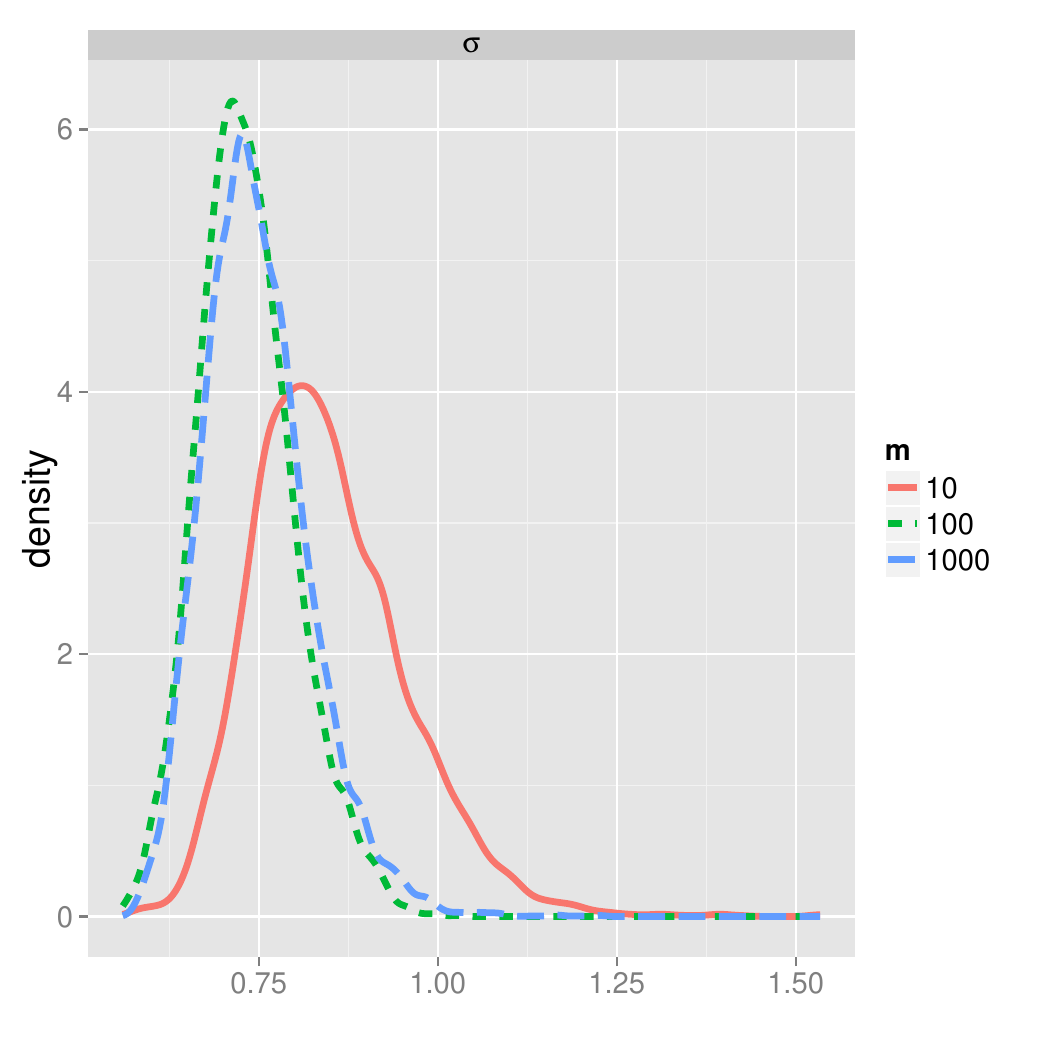}

\includegraphics[width=0.3\textwidth]{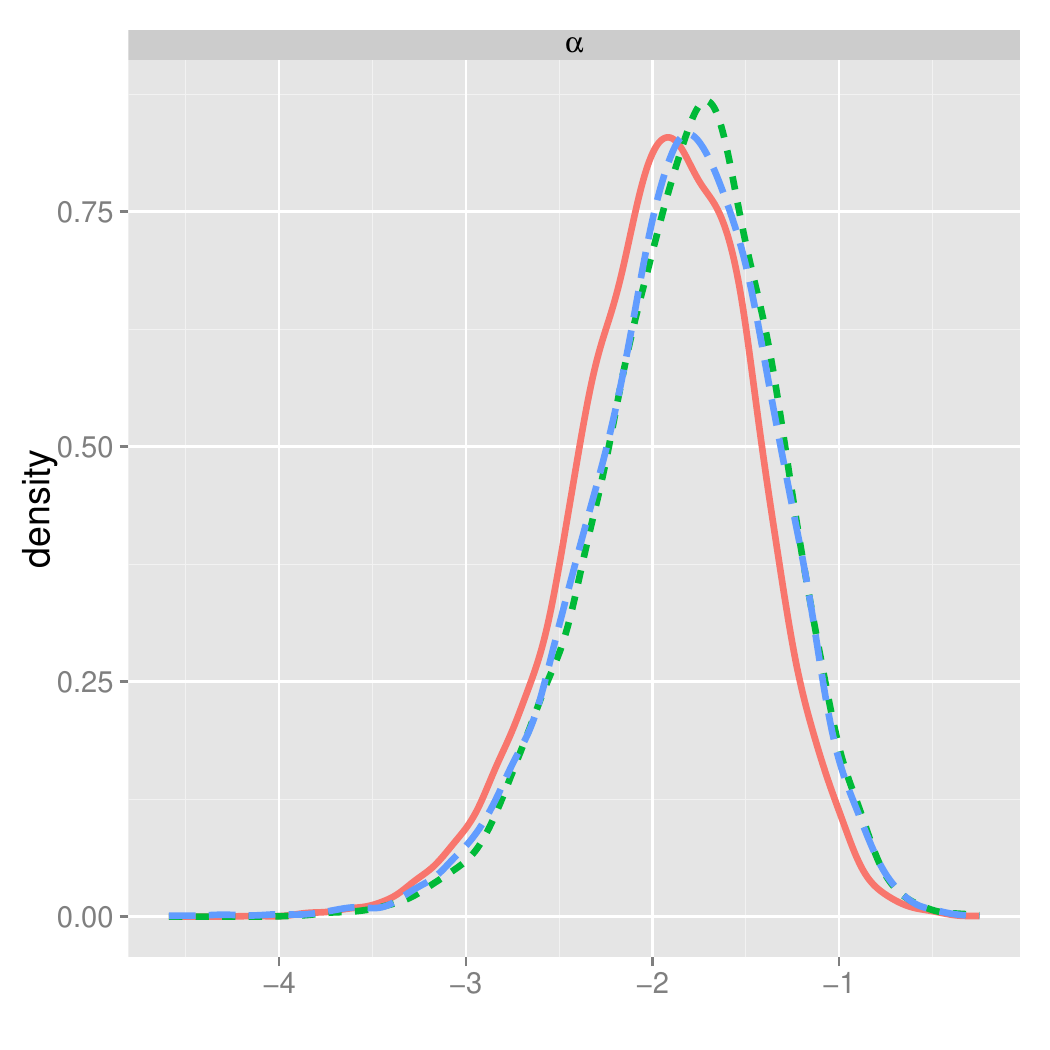}
\includegraphics[width=0.3\textwidth]{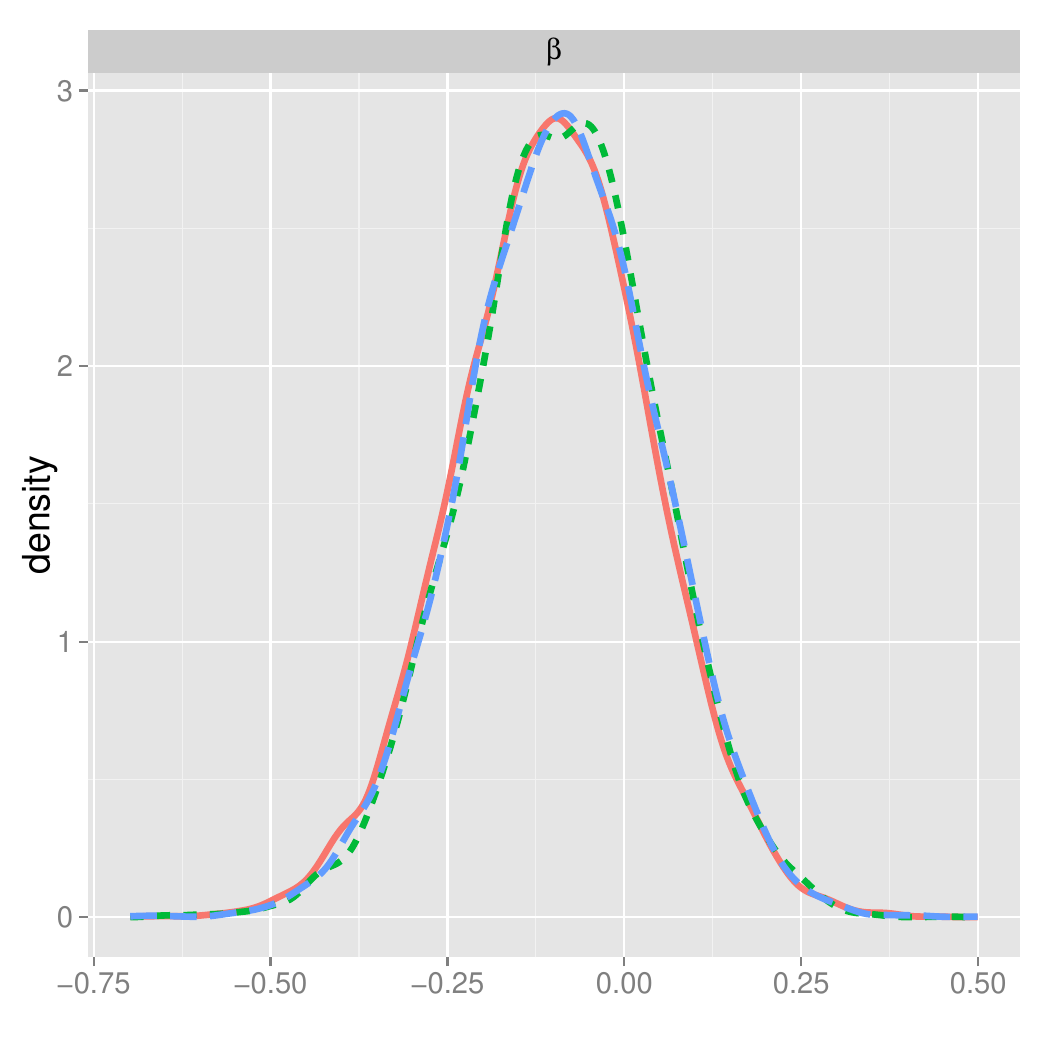}
\includegraphics[width=0.3\textwidth]{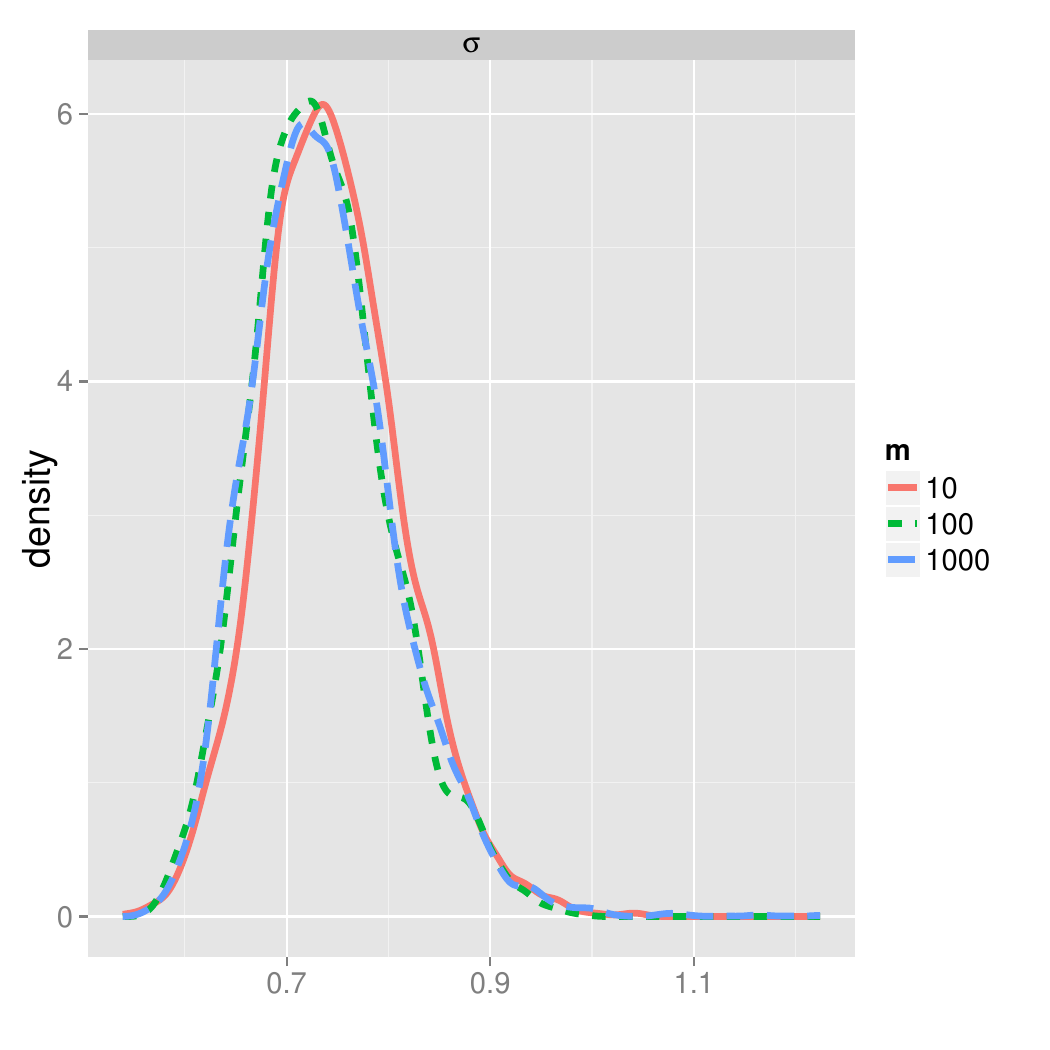}
\caption{Kernel density estimates of the parameters based on iterates 501-10.000. Top: non time-changed. Bottom: time-changed. The bias from using only a small number of imputed points ($m=10$, red curve) is clearly smaller for the time changed process. For a high number of imputed points ($m=1000$, blue dashed curve) both methods agree.}
\label{posterior1}
\end{figure}

Figures \ref{bigpanel1} and \ref{posterior1} illustrate the results of running the MCMC chain for 10.000 iterations using $m = 10, 100, 1000$ imputed points  respectively for each bridge (including endpoints), both with time change and without. Two things stand out: firstly,  increasing the number of imputed points $m$ does not worsen the mixing of the chain and secondly the vastly reduced bias when using discretisation of $U$ (especially when $m$ is small).
%---------------------

%%%%%%%%%%%%%%%%

\subsection{FitzHugh-Nagumo model}\label{sec:fitz}

The stochastic FitzHugh-Nagumo model for spike generation in squid axons is based on a two dimensional diffusion process 
with drift and diffusion coefficent parametrised as
\[ b(x)= \begin{bmatrix} \vartheta_1 (- x_1^3+x_1-x_2 + 1/2)\\
\vartheta_2 x_1- x_2 +\vartheta_3 \end{bmatrix} \qquad \si = \begin{bmatrix} \ga_1 & 0 \\ 0 & \ga_2\end{bmatrix}. \]
The first coordinate $X^{(1)}$ represents the axon membrane potential and $X^{(2)}$ is a recovery variable. Parameter estimation for the FitzHugh-Nagumo model is discussed in \cite{JensenDitlevsen} and extensively in the work of \cite{Jensen}.  
In this example we consider three type of proposals: the modified diffusion bridge (which is of Delyon-Hu type with $\la=0$), the modified diffusion bridge with random-walk type updates on the innovations and guided-proposals with random-walk type updates on the innovations. In both cases we took $\rho=0.5$ in equation \eqref{eq:rw-innovation}. 
We used time-change guided proposals as in \eqref{U} with $\tilde\sigma$ constant, $\tilde{B}\equiv 0$ and $\tilde\beta$ as in equation \eqref{eq:propbeta-simple}. This is a simple default choice.   We discretise \eqref{Usde} as follows: suppose the current iterate is $U_s$. We have
\begin{equation}\label{eq:U-dR}
\dd U_{s} =    \frac{\dd s}{T} (\tilde b(\tau(s))-b(\tau(s), \Gamma(s, U_s))) + \dd R_s
\end{equation}
with $\dd R_s= (T-s)^{-1}\Big(\I - 2a  J(s)\Big)  U_s \dd s - \sqrt{\frac{2}{T}} \frac1{\sqrt{T - s}} \sigma \dd W_s $ (where we have used  the relation $\dot{v}(s) = \tilde\beta(s)$).  
Define 
\[	\bar{R}(h)=(T-s)^{-1}\Big(\I - 2a  J(s)\Big)  U_s- \sqrt{\frac{2}{T}} \frac1{\sqrt{T - s}} \sigma  \frac{W_{s+h}-W_s}{h}. \]
To obtain an approximation $u(s+h)$ for  $U_{s+h}$ we  discretise the ordinary differential equation
\[ \dd u(s) =    \left(\frac{1}{T} \tilde b(\tau(s))-\frac{1}{T} b(\tau(s), \Gamma(s, u(s))) + \bar{R}(h) \right) \dd s, \qquad u(s) = U_s \]
 using the Runge-Kutta-4 method with step size $h$. 
We propose this discretisation scheme since by corollary \ref{cor:euler-good}, 
$\expec{R_{s+h} - R_s \mid R_s} = h \expec{\bar{R}(h) \mid R_s}$. 

We simulated the process with parameters $\vartheta_1 = 1.4$, $\vartheta_2 = 1.5$, $\vartheta_3 = 10$, $\gamma_1= 0.25, \gamma_2 = 0.2$ on the time interval from $0$ to $T=300$ using the Euler scheme with discretisation step $0.0004$ , starting in $[0\; 1]'$, retaining $400$ equidistant observations and the starting point. With these parameters this process presents a challenging estimation problem due to the strong nonlinear dynamics in the drift. 

We chose independent centred Gaussian priors with variance $50$ for the parameters $\vartheta_1, \vartheta_2, \vartheta_3$ and a product  $\operatorname{InvGamma}(0.002,0.002)$ prior on $(\ga^2_1, \ga_2^2)$.

  We used Metropolis-Hastings steps for updating $\ga_i$  by setting $\log \ga_i^\circ=\log \ga_i + 0.02 Z_i$ ($i=1,2$), where $Z_i\sim N(0,1)$. For $j=1,2, 3$ we took $\vartheta_j^\circ=\vartheta_j + \nu_j Y_j$, with $Y_j$ independent Uniform random variables on $[-1,1]$, $\nu_1=\nu_2=0.03$ and $\nu_3=0.15$. 

We estimated the joint posterior of unobserved path and parameters $\vartheta$, $\gamma$ using algorithm \ref{alg1}.
  
We ran the algorithm for the three different proposals with  $m = 10$ and $m = 25$ and $m=100$. Each simulation was stopped after $1$ hour. The simulations were done on a computer equipped with 4 core Xeon CPU clocked at 3.40GHz with 30 GiB memory.
In figure \ref{fig:fhn-trace-theta3E} trace-plots with respect to both computing time and iterate number for $\vartheta_3$ are shown for the three  samplers when $m\in \{10, 25, 100\}$. While iterates for the guided-proposals are more costly, the algorithm with these proposals does reach the stationary region way faster than the two variants of the modified diffusion bridge, especially when $m=100$. However, solely examining trace-plots for the parameters can be misleading as illustrated by figure  \ref{fig:fhn-trace-bridges-accC}. Here, we plotted the average acceptance probability for bridge proposals (on a log10-scale) for the segments in between the $99$-th and $200$-th observations (the picture is representative for all segments). At certain segments the acceptance probabilities differ by several magnitudes.  These segments  correspond precisely to observations  during an excursion from the meta-stable region. In these excursions the diffusion path follows closely the strong nonlinear drift dynamics, unlike in the meta-stable region. Small acceptance probabilities  manifest themselves in slow convergence of the chain.

In addition we ran the algorithm for a longer time ($16$ hours) with $m=200$. In this case, we simply used the Euler-approximation for the time-changed guided proposals. Actually, the  Runge-Kutta-4 method is due to the stiffness of the SDE and only necessary in case of few imputed points.    Trace-plots for $\vartheta_3$ and $\ga_1$ are shown in figure \ref{fig:fhn-trace-theta34Long}. The trace-plot of $\ga_1$ against iteration number clearly shows that the guided proposals chain mixes better.  The posterior means obtained by these methods were then considered to be the ``true'' posterior mean. These were used in computing the error values in table \ref{table:fh}.

\begin{table}
{\small
\begin{tabular}{l || c | c | c | c | c | c| c | c}\toprule
\cmidrule{1-9}
 m =  10  & RRSE &$ \vartheta_1 $ &  $ \vartheta_2 $ &  $ \vartheta_3 $ &  $ \gamma_1 $ &  $ \gamma_2 $ &  mESS & K\\
\cmidrule{1-9}
mdb-rw  &  0.3  & 0.96  & 0.95  & 0.75  & 0.94  & 0.94  &   598.4  &88485 \\ 
gp-tc  &  0.29  & 1.03  & 1.04  & 0.76  & 1.02  & 0.99  &   333.67  &27366 \\ 
mdb-ind  &  0.49  & 0.9  & 0.88  & 0.6  & 1.08  & 1.05  &   761.25  &88685 \\ 
\cmidrule{1-9}
 m =  25  & RRSE &$ \vartheta_1 $ &  $ \vartheta_2 $ &  $ \vartheta_3 $ &  $ \gamma_1 $ &  $ \gamma_2 $ &  mESS & K\\
\cmidrule{1-9}
mdb-rw  &  0.18  & 0.95  & 0.94  & 0.85  & 1.06  & 0.96  &   439.12  &47576 \\ 
gp-tc  &  0.1  & 1  & 1  & 0.91  & 1.05  & 0.99  &   177.78  &12359 \\ 
mdb-ind  &  0.42  & 0.91  & 0.89  & 0.65  & 1.16  & 1.09  &   505.77  &47006 \\ 
\cmidrule{1-9}
 m =  100  & RRSE &$ \vartheta_1 $ &  $ \vartheta_2 $ &  $ \vartheta_3 $ &  $ \gamma_1 $ &  $ \gamma_2 $ &  mESS & K\\
\cmidrule{1-9}
mdb-rw  &  0.2  & 0.94  & 0.93  & 0.84  & 1.09  & 0.99  &   216.48  &14259 \\ 
gp-tc  &  0.01  & 0.99  & 0.98  & 1.01  & 1.08  & 1.01  &   85.47  &3365 \\ 
mdb-ind  &  0.39  & 0.89  & 0.87  & 0.68  & 1.19  & 1.1  &   246.4  &14441 \\ 
\cmidrule{1-9}
\end{tabular}
}
\caption{Estimation results FitzHugh-Nagumo model based on all iterates (no burn-in was considered). For each parameter we report its relative error with respect to the posterior mean obtained from the $16$-hour run with $m=200$ imputed points using the time-changed guided proposals.  Furthermore, we report the Root Relative Squared Error (RRSE) and the multivariate Effective Sample Size (mESS) from \cite{vats}. $K$ is the total number of iterations executed within one hour.  }
\label{table:fh}
\end{table}

\begin{figure}
\centering
\includegraphics[width=\textwidth]{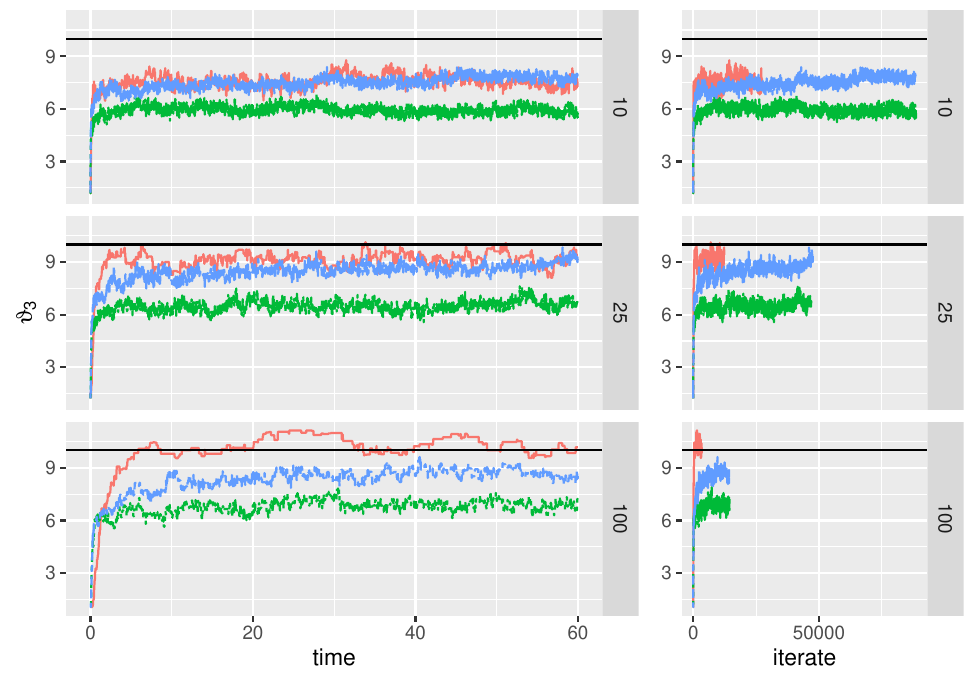}\\
\includegraphics[width=6cm]{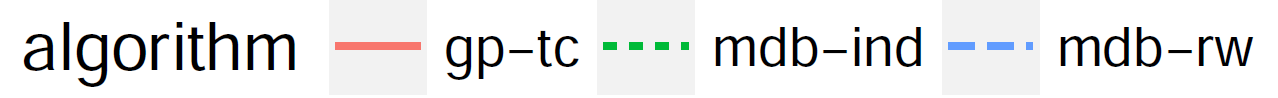}
\caption{Trace-plots for $\vartheta_3$.  Left: with respect to  computing time in minutes. Right: with respect to iterate number. The three different panels correspond to $m=10$, $m=25$ and $m=100$. ``gp-tc'' for guided proposals time-changed; ``mdb-rw'' for random-walk type modified diffusion bridges; ``mdb-rw'' for modified diffusion bridges.}
\label{fig:fhn-trace-theta3E}
\end{figure}

\begin{figure}
\centering
\includegraphics[width=\textwidth]{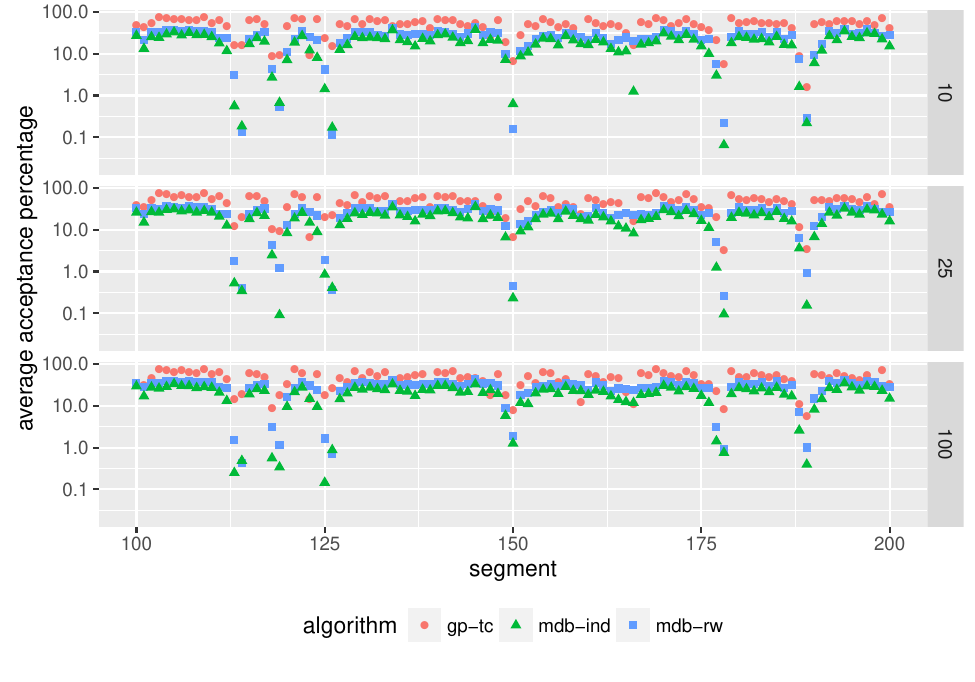}
\caption{Average acceptance percentages for proposed bridges on segments $100$ up to $200$. The three different panels correspond to $m=10$, $m=25$ and $m=100$. }
\label{fig:fhn-trace-bridges-accC}
\end{figure}

\begin{figure}
\centering
\includegraphics[width=\textwidth]{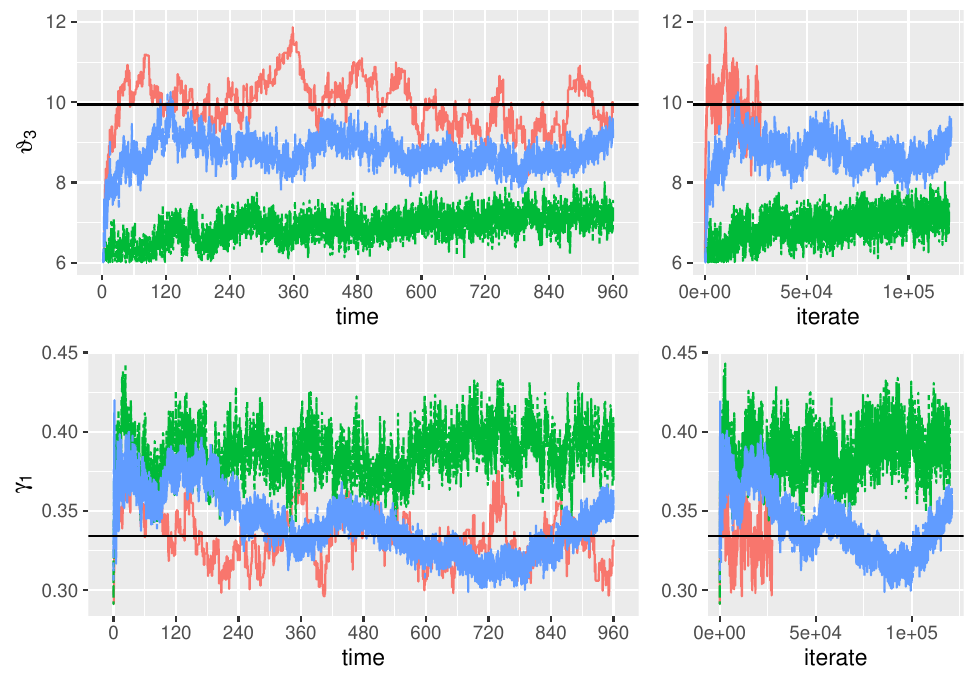}\\
\includegraphics[width=6cm]{legend.png}
\caption{Trace-plots for $\vartheta_3$ and $\ga_1$.  Left: with respect to  computing time in minutes. Right: with respect to iterate number. In all cases  $m=200$. ``gp-tc'' is guided proposals time-changed; ``mdb-rw'' is random-walk type modified diffusion bridge; ``mdb-rw'' is modified diffusion bridge. The initial $735$ iterates of $\vartheta_3$ are smaller than $6$ and not shown. }
\label{fig:fhn-trace-theta34Long}
\end{figure}

%%%%%%%%%%%%%%

%%%%%%%%%%

%%%%%%%%%%%%%%%

\appendix

\section{Proof of Lemma \ref{lem:tc}} 

For ease of notation we will write $\tau$ and $\dot \tau$ instead of $\tau(s)$ and $\dot \tau(s)$. 
If $X$ satisfies the SDE
\[ \dd X_s = b(s,X_s)\dd s + \sigma(s, X_s)\dd \tilde W_s\]
and we are given a  smooth function $\tau = \tau(s)$,  $\tau\colon[0,T) \to \RR_+$ with positive derivative $\dot \tau$, then
\begin{equation}\label{timechange} \dd X_\tau = \dot\tau b(\tau, X_\tau) \dd s + \sqrt{\dot\tau} \sigma(\tau, X_\tau) \dd  W_s,
\end{equation}
where $W$ is a different Brownian motion on the same probability space as $\tilde W$.

Applying this to $X^\circ$ (defined in equation \eqref{xcirc}) gives
\[
\dd X^{\circ}_{\tau} =2(1 - s/T)[b(\tau, X^\circ_\tau) + a(\tau, X^\circ_\tau)\tilde r(\tau, X^\circ_\tau)]\dd s + \sqrt{2(1 - s/T)} \sigma(\tau, X^\circ_\tau) \dd W_s
\]
%one has with $f(s,x) = -x/(T-s)$, $\dot f(s,x) = -(T-s)^2$ $f'(s,x) = -1/(T-s)$, $f''(s,x) = 0$,
by It\=o's formula
\begin{align*} \dd U_s &= \dd\left(\frac{v(\tau)-X^\circ_{\tau}}{T-s}\right) \\ &
%(d)/(ds)((f(s (2-s/T)))/(T-s)) = (2 (s-T)^2 f'(s (2-s/T))+T f(s (2-s/T)))/(T (s-T)^2)
= \frac2T \dot v(\tau)\dd s +\frac{v(\tau)}{(T-s)^2} \dd s
  - \frac{X^\circ_{\tau}}{(T-s)^2} \dd s\\ &\qquad \qquad - \frac2T [b(\tau, X^\circ_\tau) + a(\tau, X^\circ_\tau)\tilde r(\tau, X^\circ_\tau)] \dd s - \frac{\sqrt{2/T}}{\sqrt{T - s}} \sigma(\tau, X^\circ_\tau) \dd W_s
\\ &= \frac2T \dot v(\tau)\dd s + \frac{U_s}{T-s} \dd s
- \frac2T [b(\tau, X^\circ_\tau) + a(\tau, X^\circ_\tau)\tilde r(\tau, X^\circ_\tau)] \dd s \\& \qquad \qquad - \frac{\sqrt{2/T}}{\sqrt{T - s}} \sigma(\tau, X^\circ_\tau) \dd W_s.
\end{align*}
By equation \eqref{eq:rel-r-H},
\[
\tilde r(\tau, X^\circ_\tau) = \tilde H(\tau)(v(\tau) - X^\circ_\tau)
= \tilde H(\tau)(T-s)U_s = J(s) \frac{T}{T-s} U_s.
\]
The result now follows from substituting this expression and using the relation $X^\circ_{\tau(s)} = \Gamma(s, U_s)$ (see \eqref{eq:Xcirc-back}).
 
The statement on $J$ is a consequence of $\tilde{a}(s) J(s)=\tilde{H}(\tau(s))(T-\tau(s))$ together with 
$\lim_{t\to T}\tilde H(t)(T-t) = \tilde a^{-1}$ (see \cite[Lemma 8]{Schauer}). 

The expression for the integral follows upon the substitution $s:= \tau(s)$ and using relation \eqref{eq:rel-r-H}. 

\section{Proof of Theorem \ref{thm:tcou}}\label{app:proofs}

\begin{proof}
By straightforward calculus, the process $U$ satisfies a SDE with drift coefficient 
\[
m \dot{v}(\tau) \dot\tau + \dot{m}m^{-1} U_s -m \dot\tau \left[b(\tau,\Gamma)+a(\tau,\Gamma) \tilde{H}(\tau) m^{-1} U_s\right]
\]
and diffusion coefficient as given in the theorem. We can rewrite the drift coefficient using specific properties of $m$. 

For the first term in the drift, note that
\[
\tilde{b}(\tau, \Gamma)  = \tilde{B}(\tau) \left(v(\tau)-m^{-1} U_s\right) + \tilde\beta(\tau)  = \dot{v}(\tau) - \tilde{B}(\tau) m^{-1} U_s,
\]
where we have used  the relation $\dot{v}(s)=\tilde{B}(s) v(s) + \tilde\beta(s)$ at the second equality.  Multiplying by $m\dot\tau$ we get
\begin{equation}\label{eq:eersteterm}  m \dot\tau \dot{v}(\tau) = m \dot\tau 	\tilde{b}(\tau, \Gamma)  +m\dot\tau \tilde{B}(\tau) m^{-1} U_s. \end{equation}

Next, we rewrite the second term appearing in the drift. Using $\mathrm{d} {A}^{-1}/\mathrm{d}t = - {A}^{-1}\left( \mathrm{d}{A}/\mathrm{d}t\right) {A}^{-1}$  for an invertible matrix $A$, we obtain that
\begin{equation}\label{eq:tweedeterm} \dot{m} m^{-1} =-(T-s)^{-1} \I -\dot\tau \left(\tilde{B}(\tau)-\tilde{a}(\tau) \tilde{H}(\tau)\right). \end{equation}

Substituting \eqref{eq:eersteterm} and \eqref{eq:tweedeterm} into the drift and reordering terms shows that the drift of $U$ equals
\begin{multline*} -(T-s)^{-1} U_s - m\dot\tau \left(b(\tau,\Gamma)-\tilde{b}(\tau,\Gamma)\right)\\ +m\dot\tau \left(\tilde{B}(\tau)-a(\tau,\Gamma) \tilde{H}(\tau)\right) m^{-1}U_s - \dot\tau \left(\tilde{B}(\tau)-\tilde{a}(\tau)\tilde{H}(\tau)\right) U_s. \end{multline*}
This can be simplified to the form given in the theorem by using that
\[  \tilde{B}(\tau)-\tilde{a}(\tau) \tilde{H}(\tau)= m\left(\tilde{B}(\tau)-\tilde{a}(\tau) \tilde{H}(\tau)\right) m^{-1} \]
which follows from the defining relation of $F^\star$. 
\end{proof}

\bibliographystyle{harry}

\bibliography{lit}

\end{document}